\documentclass[final,onefignum,onetabnum,oneeqnum,onethmnum]{siamart171218}

\usepackage{lipsum}
\usepackage{amsfonts}
\usepackage{graphicx}
\usepackage{epstopdf}
\usepackage{algorithmic}
\ifpdf
  \DeclareGraphicsExtensions{.eps,.pdf,.png,.jpg}
\else
  \DeclareGraphicsExtensions{.eps}
\fi

\usepackage{amssymb}
\usepackage{gensymb}
\usepackage{enumerate}
\usepackage{mathtools}
\DeclarePairedDelimiter{\ceil}{\lceil}{\rceil}
\usepackage{tikz}

\usepackage[caption=false]{subfig} 

\usepackage[titletoc,toc,title]{appendix}

\newsiamremark{remark}{Remark}
\newsiamremark{hypothesis}{Hypothesis}
\crefname{hypothesis}{Hypothesis}{Hypotheses}
\newsiamthm{claim}{Claim}

\headers{Prob. Gathering with Simple Sensors}{A. Barel, T. Dag\`es, R. Manor, and A. M. Bruckstein}

\title{Probabilistic Gathering of Agents with Simple Sensors\thanks{Submitted to the editors April 12, 2020.
}}

\author{Ariel Barel\thanks{Department of Computer Science, Technion - Israel Institute of Technology, Haifa, ISRAEL 
  (\email{arielba@technion.ac.il}, \url{http://www.cs.technion.ac.il/people/arielba/}
  ).}
\and Thomas Dag\`es\footnotemark[2] \and Rotem Manor\footnotemark[2] \and Alfred M. Bruckstein\footnotemark[2]}

\usepackage{amsopn}

\makeatletter
\newcommand*{\addFileDependency}[1]{
  \typeout{(#1)}
  \@addtofilelist{#1}
  \IfFileExists{#1}{}{\typeout{No file #1.}}
}
\makeatother

\newcommand*{\myexternaldocument}[1]{%
    \externaldocument{#1}%
    \addFileDependency{#1.tex}%
    \addFileDependency{#1.aux}%
}

\ifpdf
\hypersetup{
  pdftitle={Probabilistic Gathering of Agents with Simple Sensors},
  pdfauthor={A. Barel, T. Dag\`es, R. Manor, and A. M. Bruckstein}
}
\fi

\myexternaldocument{ex_supplement}

\begin{document}
\maketitle

\begin{abstract}
  Gathering is a fundamental task for multi-agent systems and the problem has been studied under various assumptions on the sensing capabilities of mobile agents. This paper addresses the problem for a group of agents that are identical and indistinguishable, oblivious, and lack the capacity of direct communication. At the beginning of unit-time intervals, the agents select random headings in the plane and then detect the presence of other agents behind them. Then they move forward only if no agents are detected in their sensing ``back half-plane''. Two types of motion are considered: when no peers are detected behind them, either the agents perform unit jumps forward, or they start to move with unit speed while continuously sensing their back half-plane, and stop whenever another agent appears there. For the first type of motion extensive empirical evidence suggests that with high probability clustering occurs in finite expected time to a small region with diameter of about the size of the unit jump, while for continuous sensing and motion we can prove gathering in finite expected time if a ``blind-zone'' is assumed in their sensing half-plane. Relationships between the number of agents or the size of the blind-zone and convergence time are empirically studied and compared to a theoretical upper-bound dependent on these factors.
\end{abstract}



\section{Introduction}
This paper deals with gathering of multi-agent systems, based on a simple decentralized control law. Agents move according to local information provided by their sensors. The agents are assumed to be identical and indistinguishable, memoryless (oblivious), with no explicit communication between them. They do not have a common frame of reference (i.e. agents are not equipped with GPS sensors or compasses). This paper is an extended and improved version of a preprint that appeared on arXiv in 2018 \cite{barel2018previousSimple}.

A wealth of gathering algorithms were described and analyzed in the multi-agent robotics literature. They differ in the assumptions made on the sensing that is performed by the agents, in the assumptions on the possible moves that can be made and the computational requirements for the decision process that leads to the response \cite{bruckstein1991ants,suzuki1999distributed,ando1999,jadbabaie2003,gazi2003stability,gordon2004,gordon2005,olfati2006flocking,ji2007,olfati2007consensus,gordon2008,bellaiche2015,manor2016Chase}. 

A recent report \cite{barel2019come} surveyed in detail gathering or clustering algorithms under the assumptions of limited vs. unlimited visibility sensing, complete relative position sensing vs. bearing only information provided by the sensors, and discrete vs. continuous motion schedules for the agents. We here present a novel gathering, or geometric consensus algorithm based on a simple randomized rule of motion for agents with headings that can only sense the presence or absence of other agents in a half-plane behind them. We assume that the agents have intrinsic (but not global) orientation and they can move forward, but each agent can, at various instances, select a new heading at random. This is accomplished by doing an independent and uniformly distributed turn over $[0,2\pi]$ from their current orientation.

Under these assumptions, we consider two different gathering algorithms. The first one assumes that agents make a forward jump of size $1$ whenever there are no agents behind them, i.e. in their back half-plane. The agents act synchronously and new headings are selected at random and independently at each unit-time, then their ``back half-plane'' sensor's reading tells the agents whether to jump forward or stay put. Extensive experimental results with this process shows that probabilistic gathering to a small region occurs in time proportional to the number of agents. We cannot prove this result yet.

The second gathering algorithm we discuss is a continuous version of this process. Here we assume that the sensing is continuously done, and the forward motion of agents during each unit time-interval is continuously conditioned on the absence of agents in the sensing area behind them. As we shall see, in this case we also need to assume a blind-zone for the backwards sensing, in order to avoid dead-lock situations preventing gathering. Under these new assumptions we can prove gathering in finite expected time, and obtain a bound on the gathering time dependent on the number of agents, the initial spread of the group, and the size of the blind-zone in sensing.

Extensive experiments are performed for both gathering algorithms. Results not only confirm our theory but also provide further empirical insight into the behaviour of the swarms. The code has been made public\footnote{\url{https://github.com/Tommoo/SwarmGatheringHalfPlanes}}.


\section{The discrete algorithm}
Consider a system of $n$ identical, anonymous, and memoryless agents specified by their time varying locations in the euclidean plane $\{p_i(k)\}_{i=1,2,...,n} \in  \mathbb{R}^2$ and heading vectors $\{\hat\theta_i(k)\}_{i=1,2,...,n}$ which are unit vectors randomly selected on the unit circle. These quantities are unknown to the agents themselves as they lack global position and orientation information. We define ``heading" as the direction where the agent's nose is pointing, i.e. its current direction of (possible) motion. The agents implicitly interact with each other in such a way that an agent's next position after one time unit $p_i(k+1)$ is determined by the constellation of all the agents in the system.

\subsection{Sensing}
Each agent is equipped with an onboard sensing device, aimed in the opposite direction to the agent's heading $\hat{\theta}_i(t)$, covering the back half-plane (with $180\degree$ field of view). We may call this sensing device a ``Backward Looking Binary Sensor". If there is no other agent in the field of view of agent $i$, the output signal is $s_i(k) = 1$, else $s_i(k) = 0$.

\subsection{Timing and the Motion Law}\label{DynamicLawWording}
At time $k=0$ the agents are in an arbitrary initial constellation with randomly selected headings, and perform forward jumps if their sensor reading is $1$. Then at each time-step $k$ every agent changes its heading direction by choosing a uniformly distributed random turn from $[0,2\pi]$ resulting in a new random heading $\theta_i(k)$, and then, if its back closed half-plane is empty, it jumps forward a fixed step-size $d=1$. Otherwise it stays put until the next time-step.

The dynamic motion law is formally described as follows. Denote the locations of the $n$ agents at time-step $k$ by $\{p_1(k), p_2(k), ..., p_n(k)\}$. Then:
\begin{align}
    \begin{split}
        p_i(k+1) &= p_i(k) + \begin{bmatrix} \cos(\theta_i(k)) \\ \sin (\theta_i(k)) \end{bmatrix}  s_i(k)\\
        with\quad 
       	s_i(k) &= 
       	    \begin{cases}
       	        0, & \exists~j : \hat{\theta}^\intercal_i(k)[p_j(k)-p_i(k)] \leq 0\\
       	        1, & otherwise
       	    \end{cases} 
    \end{split}
    \label{eq:Dynamics}
\end{align}
is the binary output from the sensor as seen in \cref{fig:PositiveNegative}, and $\hat{\theta}_i(k)=\bigl[\begin{smallmatrix}\cos({\theta_i(k)}) \\ \sin({\theta_i(k)})\end{smallmatrix}\bigr]$ is the agent's random heading.

\begin{figure}[tbhp]
  \centering
    \includegraphics[width=70mm]{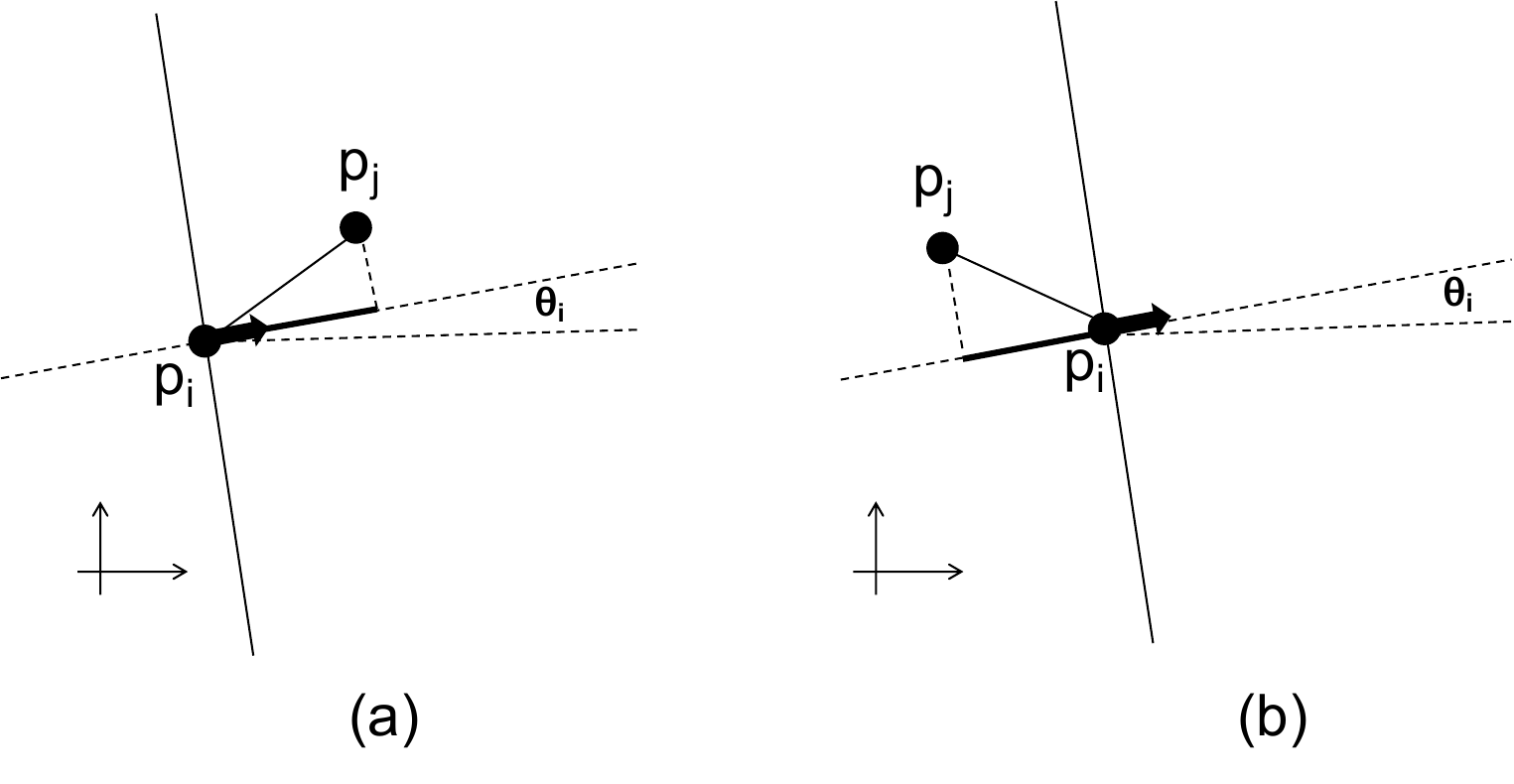}
    \caption{Sensing geometry: The heading of an agent is presented by an arrow, and the product $ \hat{\theta}_i^\intercal[p_j-p_i]$ is marked in bold line. If (a) agent $j$ is located in front of agent $i$ the product is positive, and if (b) agent $j$ is in its back the product is negative. If all products are positive $s_i(k)=1$ otherwise $s_i(k)=0$.}
    \label{fig:PositiveNegative}
\end{figure}

\subsection{Simulation Results and System Behavior}

Typical simulation results of gathering are shown in \cref{fig:discrete evolution}. In all the simulations we ran, a number of agents are randomly placed in the plane, and in all cases the system converges to an area within a circular region of radius less than $1$, whose ``center" wanders at random in the plane. Empirically, once we have reached convergence, this radius seems to randomly fluctuate around the value $\frac{1}{2}$ with variance decreasing with the number of agents. Furthermore, once convergence has been achieved, any departure from the unit disk implies a radius lower than $2$ due to the rules of motion, and empirically it decreases again to go back to convergence within a disk of radius at most $1$. This phenomenon seldom happens and only on systems with very few agents, typically less than $5$.

We found that, for estimating the expected time of convergence, a choice of $1{,}000$ runs provides a reliable estimator. A detailed analysis of this choice is presented as supplementary material in \cref{an subsec: Discrete analysis nb runs}. \Cref{ConvergenceTime_VS_n} summarizes $1{,}000$ simulation runs with different number of agents, spread uniformly over the same initial area. Notice that the effect of the number of agents on the convergence time of the system is linear, and we found a slope of approximately $20$ unit time-steps per number of agents.

\begin{figure}[tbhp]
    \centering
    \subfloat[$n=5$]{\label{fig:qty5stp1Range50}\includegraphics[width=0.45\textwidth]{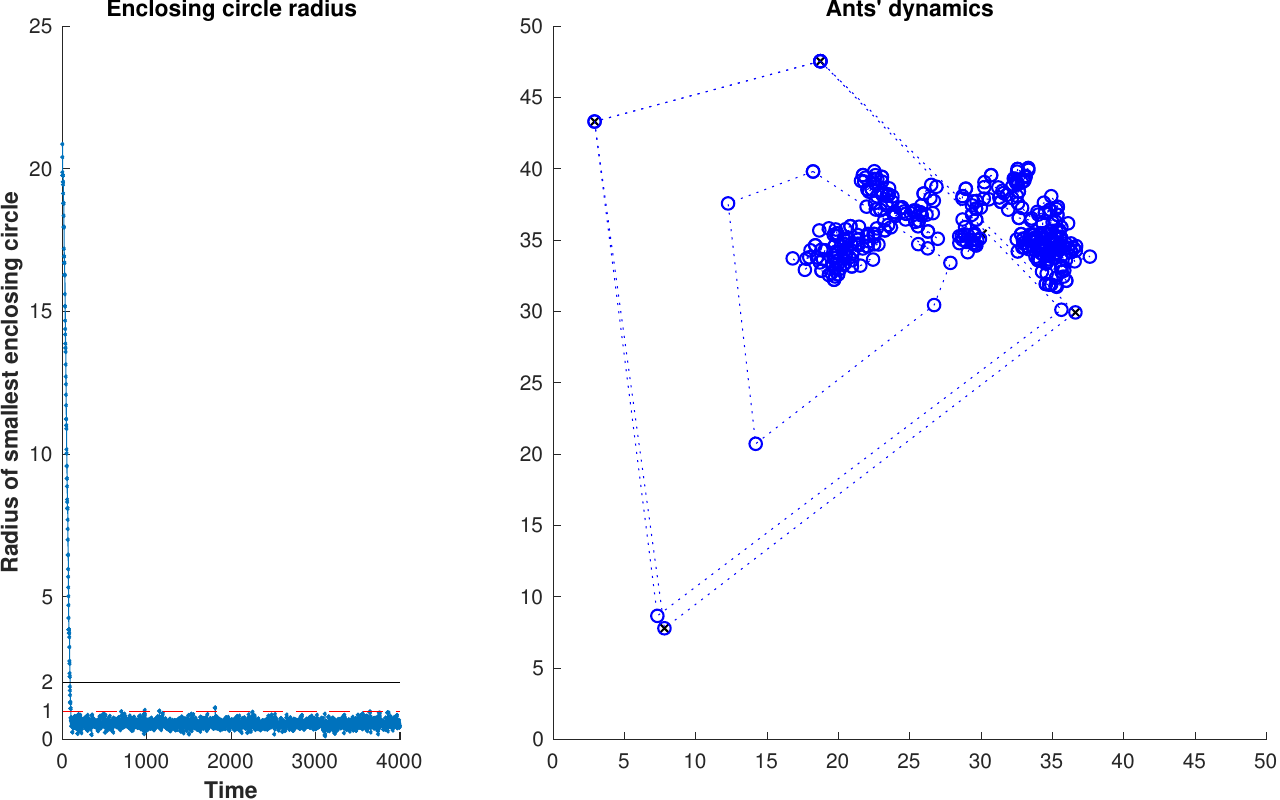}}
    \subfloat[$n=50$]{\label{fig:qty50stp1Range50}\includegraphics[width=0.45\textwidth]{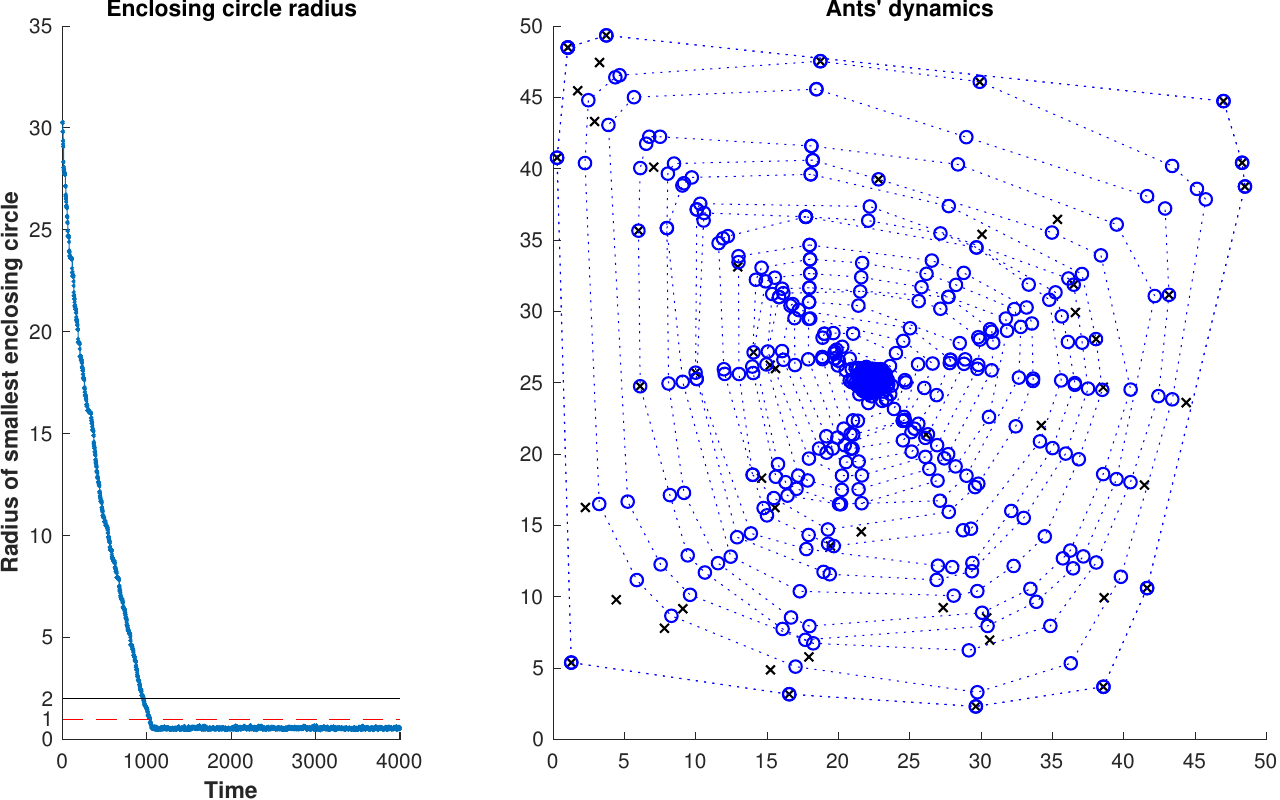}}\\    
    \subfloat[$n=150$]{\label{fig:qty150stp1Range50}\includegraphics[width=0.45\textwidth]{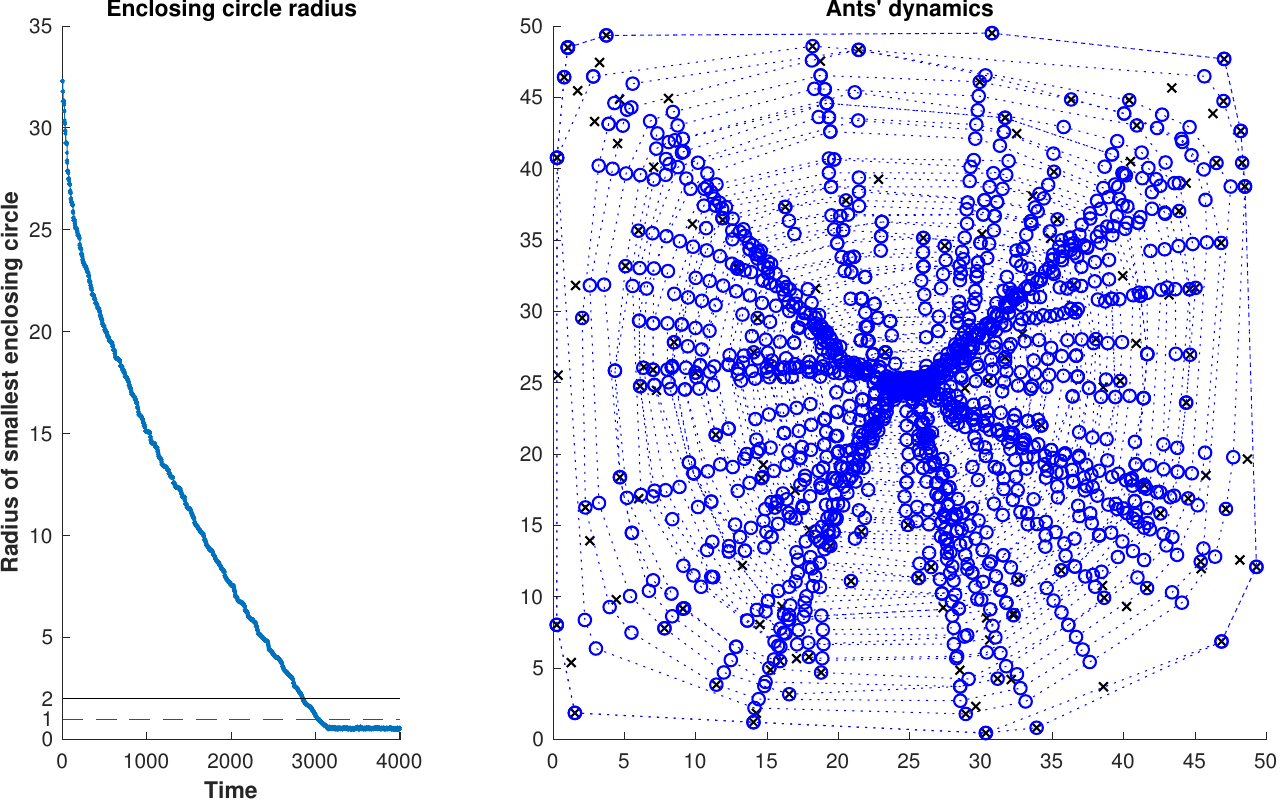}}
    \caption{Simulation results on $n\in\{5,50,150\}$ agents for the discrete dynamics algorithm, with initial random spread in a $50$ by $50$ area and step-size $1$. On the right part of each figure,
    the convex-hull of the system and its associated agents are printed every $50$ time-steps (marked with dashed lines and $\circ$). The initial position of the agents is marked with $\times$. The curve on the left part of each figure shows the decrease in the radius of the smallest enclosing circle of the agents constellation.}
    \label{fig:discrete evolution}
\end{figure}






\begin{figure}[tbhp]
  \centering
    \includegraphics[width=85mm]{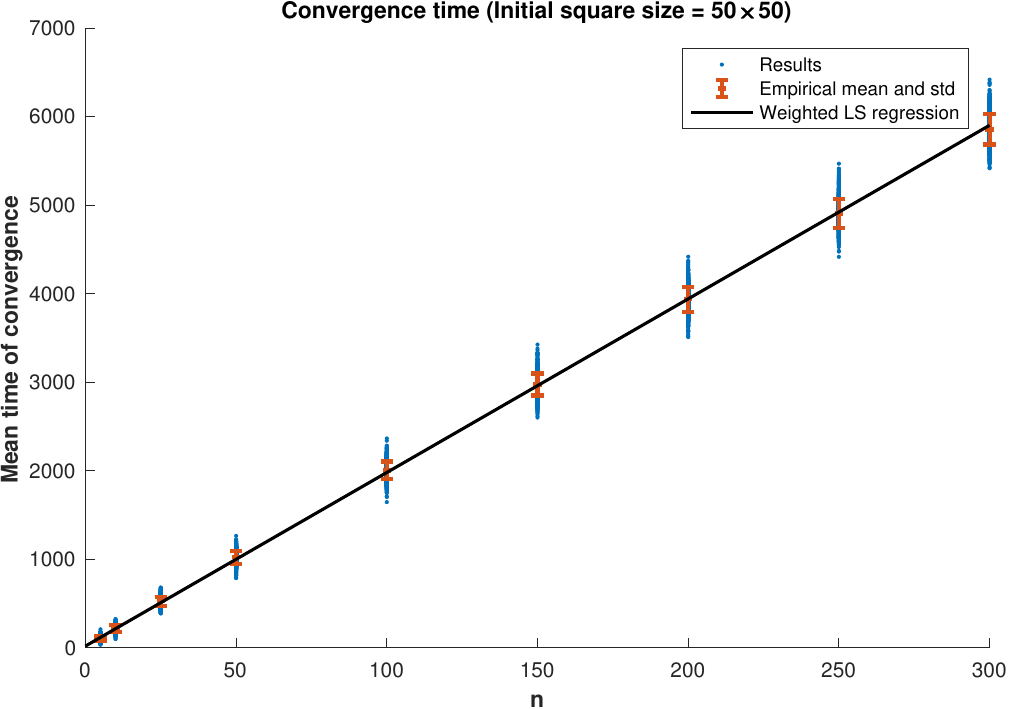}
    \caption{Convergence time vs. number of agents. All simulations were set to an initial random spread of $50$ by $50$ area and step-size $1$. Convergence time was taken for the first time when all agents were gathered in a circle of radius $1$. This process ran for $1{,}000$ repetitions with different random initial constellations. We superimpose on all results the empirical mean and standard deviation of the results. The linear graph was obtained using weighted linear least squares fitting to the average results, with weights equal to the squared inverse of the empirical standard deviation of the data to take into account the non uniform variance of the results.}
      \label{ConvergenceTime_VS_n}
\end{figure}

As shown in \cref{EmptyNonEmptyBack}, only the agents occupying the corners of the convex-hull of the constellation may select an orientation with empty back half-planes, thus only they may jump, while the inner agents necessarily stay put until they become external. Since agents' headings change randomly, agents on the convex-hull of the system have a high probability to jump towards all other agents, hence the system tends to gather, as shown in the simulations.
\begin{figure}[tbhp]
  \centering
    \includegraphics[width=45mm]{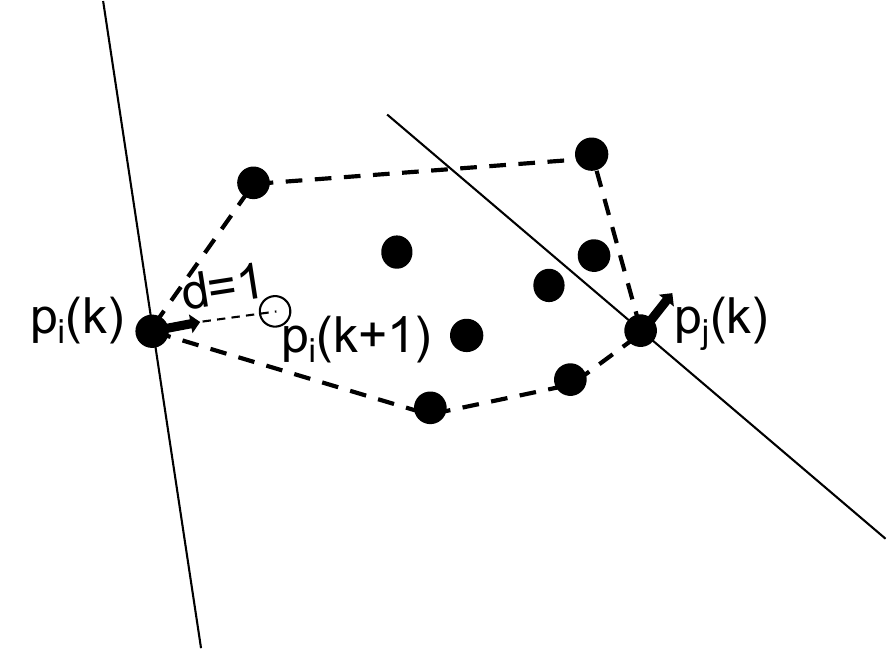}
    \caption{The headings of agents $i$ and $j$ are shown by arrows. The back half-plane of agent $i$ is currently empty of other agents, therefore it jumps a fixed step-size $d=1$ forward, while the back of agent $j$ contains some agents so it stays put, and all internal agents are guaranteed to stay put.}
      \label{EmptyNonEmptyBack}
\end{figure}

But this simplified account does not take into consideration the fact that even divergence may occur due to blind unit jumps, however with negligibly small probability. An adversarial argument clearly illustrates this. Indeed, consider a system of two agents, $n=2$, and assume that the agents' headings are (almost) perpendicular to the line they define and oriented in opposite directions (so that we have  $0<\hat{\theta}^\intercal_i(k)  [p_j(k)-p_i(k)] \ll1$ and $\hat{\theta}_j(k)= -\hat{\theta}_i(k)$). Since both their back half-planes are empty they both jump forward, and obviously we have that the distance between them increases as a result of the move.

We next suggest a modified gathering algorithm with continuous-time dynamics, for which we can prove gathering to a very small region in finite expected time.

\section{Continuous-time dynamics} 
In this system, once in a unit time-interval $\Delta t=1$, simultaneously, each agent changes its heading direction $\hat{\theta}_i(t)$ by choosing a uniformly distributed random angle between $0$ and $2\pi$. Here too, the sensor is aimed backwards (at $-\hat{\theta}_i(t)$) but it includes a ``blind-zone" half disc area of radius $\delta$, typically $\delta\ll1$ (see \cref{UL_FiniteVisibilityLyapunov}). During the time-interval $\Delta t$ an agent keeps its heading direction, and if its sensing area is empty it moves forward with a fixed velocity $v=1$, otherwise it stops. Note that we assume that the agent may move and stop during $\Delta t$ according to the changing constellation of the system. 

Denote by $d_{ij}(t)=\|p_i(t)-p_j(t)\|$ the distance between agents $i$ and $j$ at time $t$, so that if $d_{ij}<\delta$, the agents are too close to potentially see each other, otherwise we call them ``separated". The dynamic law in continuous-time is:
\begin{align}
    \begin{split}
        \dot{p}_i(t) &= \begin{bmatrix} \cos(\theta_i(t)) \\ \sin (\theta_i(t)) \end{bmatrix}  s_i(t) \quad and \quad \theta_i(t)=\sum_{k=1}^{\infty}\chi^{(i)}_k1_{\Delta_k}(t) \\
        where \quad &\chi^{(i)}_k \text{are iid uniformly distributed over } [0,2\pi] \\
        &1_{\Delta_k}(t)= 
            \begin{cases}
                1, &\quad  \text{for } t \in [k, k+1)\\
                0, &\quad otherwise
            \end{cases} \\
        &s_i(t)= 
            \begin{cases}
       	        0, &\quad \exists~j \;:\; d_{ij}>\delta \text{ and } \hat{\theta}^\intercal_i(t)[p_j(t)-p_i(t)] \leq 0\\
	            1, &\quad otherwise
            \end{cases}
    \end{split}
    \label{eq:continuous dynamics}
\end{align}

For the agents' dynamics, by abuse of notation, derivatives actually represent well-defined right derivatives. In the following proofs we often omit the time index $t$.

\begin{lemma}\label{dijNeverIncrease}
The distance between two ``separated" agents never increases.
\end{lemma}

\begin{proof}
Suppose agents $i$ and $j$ are ``separated" at time $t$ so that $d_{ij}>\delta$. Denote by $\theta_{ij}$ the (current) small angle between vector $p_j-p_i$ and the heading direction of agent $i$ as shown in \cref{UL_FiniteVisibilityLyapunov}.

\begin{figure}[tbhp]
  \centering
 \includegraphics[width=50mm]{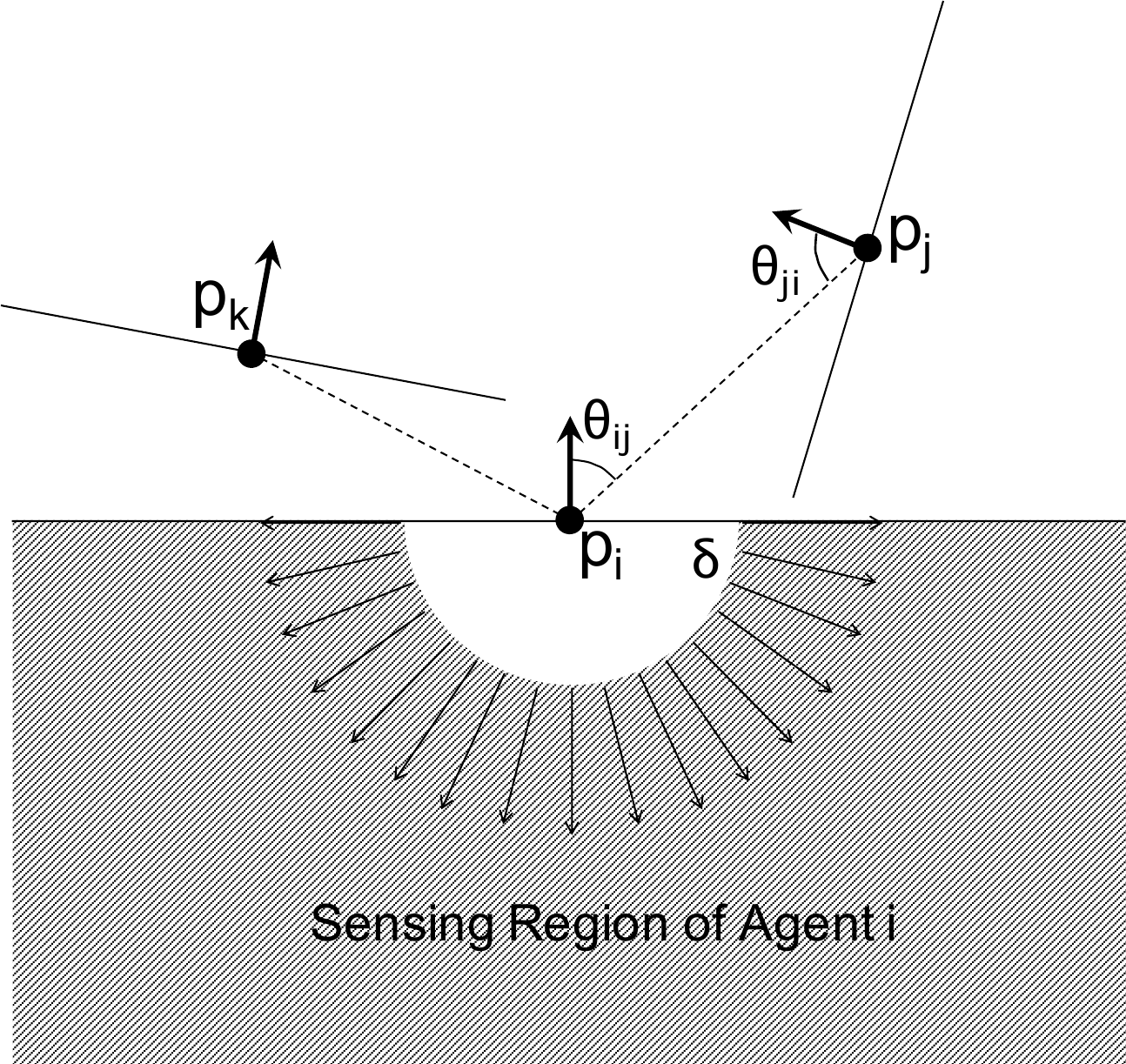}
    \caption{The dashed half-plane region missing the half-disc of radius $\delta$ centered at $p_i$ is the sensing coverage area of agent $i$ with its dead-zone of radius $\delta$. Here the sensing region of agent $i$ is empty, therefore agent $i$ moves. Agent $j$ can also move, but agent $k$ cannot move as its sensor coverage area contains agent $i$.}
      \label{UL_FiniteVisibilityLyapunov}
\end{figure}

The derivative of the distance between $p_i$ and $p_j$ (that is the inverse of their approach speed) is given by
\begin{align}
    \label{AgentSpeed}
    \begin{split}
        \frac{d}{dt}d_{ij}&=\frac{d}{dt}\|p_j-p_i\|=-\left(\dot p_i \cdot \frac{p_j-p_i}{\|p_j-p_i\|}+ \dot p_j \cdot \frac{p_i-p_j}{\|p_i-p_j\|}\right)\\
        &=-(\|\dot{p}_i\|\cos\theta_{ij}+\|\dot{p}_j\|\cos\theta_{ji})=-(s_i\cos\theta_{ij}+s_j\cos\theta_{ji}).
    \end{split}
\end{align}

By the dynamic law, if the sensor coverage area of agent $i$ is not empty (i.e. $s_i=0$) it does not move. In this case its speed is $\|\dot{p}_i\|=s_i=0$. Otherwise, all other agents including agent $j$ are either in front of it or in its ``blind-zone'' and then it moves forward at $\|\dot{p}_i\|=s_i=1$. However, since we assumed agents $i$ and $j$ to be ``separated", agent $j$ is then in this case necessarily in front of agent $i$, thus  $-\frac{\pi}{2} < \theta_{ij} < \frac{\pi}{2}$, i.e. $0<\cos\theta_{ij}\leq 1$.  Similar arguments hold for agent $p_j$, that is either $\|\dot{p}_j\|=0$ or $\|\dot{p}_j\|=v=1$ and  $0<\cos \theta_{ji}\leq 1$. Hence we have that
\begin{equation}\label{NeverLoseFriend}
d_{ij}> \delta \implies \frac{d}{dt}d_{ij} \leq 0 .
\end{equation}

\end{proof}

\begin{corollary}\label{NeverLoseFriendCorollary}
Agents within range $\delta$ at time $t$ remain within range $\delta$ at all $t'\geq t$.
\end {corollary}
\begin{proof}
This is a direct consequence of \cref{dijNeverIncrease,NeverLoseFriend}. Note that if an agent $j$ is closer than $\delta$ to agent $i$ (so that $d_{ij}<\delta$), in order for $d_{ij}$ to increase above $\delta$, it first has to reach the value $d_{ij}=\delta$ since when an agent moves, it moves in a continuous motion, but by \cref{dijNeverIncrease} the value $d_{ij}=\delta$ cannot increase. A detailed proof for this corollary is given in \cref{App: Proof Bourbaki Never Lose Friend}.
\end{proof}

Next we show that if not all agents are confined inside a circle of radius $\delta$ there is a strictly positive probability bounded away from zero for the distance between pairs of agents to decrease at a positive and bounded away from zero rate, until all agents are confined in a circle of radius $\delta$.

\subsection{Proof of Convergence}

\begin{theorem}\label{Theorem} Continuous dynamics with agents acting according to the motion law given in \cref{eq:continuous dynamics}, converges to a region of radius $\delta$ in finite expected time.
\end{theorem}

\begin{proof}
We know by \cref{NeverLoseFriendCorollary} that pairs of agents within range $\delta$ from each other at some time $t$ will remain so forever. We shall next show that while in the agents' configuration there exists pairs of agents at distance bigger than $\delta$ from each other, there is a significant (bounded away from zero by a constant) probability that there will be a significant decrease in the distance between them.

\subsubsection{Preliminaries}\ 

\begin{lemma}\label{lem:UL_SharpestAngle}
As long as not all agents are confined in a circle of radius $\delta$, there is an agent with a strictly positive probability bounded away from zero by a constant to move a distance bounded away from zero by another constant during $\Delta t$.
\end{lemma}

\begin{proof}
The sum of corner angles of any convex polygon of $m$ corners is given by $\pi(m-2)$, therefore the sharpest corner of a convex polygon of $m$ corners is bounded from above by $ \frac{\pi(m-2)}{m}=\pi(1-\frac{2}{m})$. Since the maximal number of corners of the convex-hull of a system of $n$ agents is $n$, we have that $\alpha_s$, the sharpest corner of the convex-hull of a system on $n$ agents, is bounded by
\begin{equation}\label{UL_SharpestAngleValue}
\alpha_s\leq \pi(1-\frac{2}{n}).
\end{equation}

Let $\alpha_i=\angle p_{i-1}p_ip_{i+1}$ be the convex-hull angle at a corner $p_i$, and let $p_{i-1}$ and $p_{i+1}$ be the locations of the corners of the convex-hull adjacent to $p_i$ (see \cref{fig:UL_SharpestAngle}). 

\begin{figure}[tbhp]
  \centering
   \includegraphics[width=0.6\textwidth]{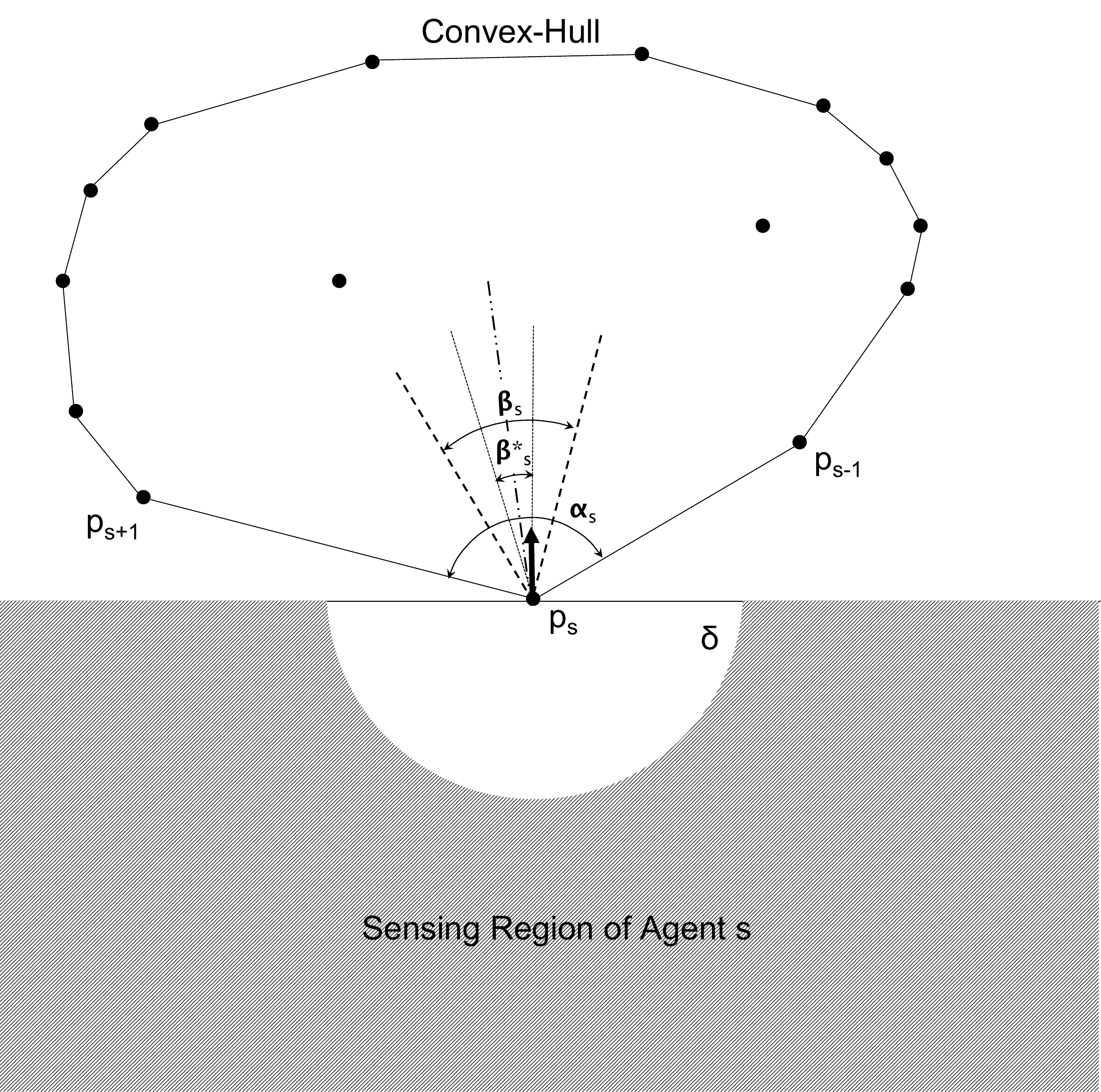}
    \caption{Agent $s$ at the sharpest corner of the convex-hull is shown with its sensing area. Agents $p_{s-1}$ and $p_{s+1}$ are the adjacent convex-hull corners to $p_s$. The sides of $\beta_s$ are perpendicular to those of $\alpha_s$, therefore $\beta_s=\pi-\alpha_s$. Angle $\beta^*_s=\frac{1}{2}\beta_s$ share the same bisector with $\alpha_s$ and $\beta_s$. The black arrow shows the selected heading direction of agent $s$.}
      \label{fig:UL_SharpestAngle}
\end{figure}

Denote by $\beta_i=\pi-\alpha_i$ the smaller angle between the two lines perpendicular to the sides of $\alpha_i$ so that if the random heading of agent $i$ falls inside $\beta_i$ it can move. Note that if the heading of $p_i$ is precisely along one of the sides of $\beta_i$ (and not explicitly inside $\beta_i$), agent $i$ may not move (as then $p_{i-1}$ or $p_{i+1}$ may be located in its sensing area). Therefore consider the symmetric central half of $\beta_i$ about the bisector of $\beta_i$ (and $\alpha_i$ too) denoted by $\beta^*_i= \frac{1}{2}\beta_i=\frac{1}{2}(\pi-\alpha_i)$. If the heading of agent $i$ on the convex hull falls inside $\beta_i^*$, then agent $i$ is guaranteed to move at the beginning of the time-interval. The probability for this event to happen is $\frac{1}{2\pi}\beta^*_i=\frac{1}{2\pi}\frac{1}{2}\beta_i = \frac{1}{4\pi}(\pi-\alpha_i)$. Hence the probability for agent $i$ (defining the convex-hull) to move is lower-bounded as follows: 
$$Pr(\textrm{agent } i \textrm{ moves})\ge\frac{1}{4\pi}(\pi-\alpha_i).$$

To further bound this probability independently of the constellation, let us consider the agent at the sharpest corner of the convex-hull, so that by \cref{UL_SharpestAngleValue} we know that $\alpha_s\leq \pi(1-\frac{2}{n})$. Hence the probability of $s$, the agent at the sharpest corner of the convex-hull $p_s$ to move is lower-bounded by
\begin{equation}\label{UL_SharpestAngleToMove}
Pr(\textrm{agent } s \textrm{ moves})\ge\frac{1}{4\pi}(\pi-\pi(1-\frac{2}{n}))=\frac{1}{2n}\ .
\end{equation}
Let us define the case where the heading of agent $s$ is inside its associated $\beta^*_s$ (at the beginning of a time-interval) as a ``successful time-interval", and its associated bound \cref{UL_SharpestAngleToMove} as ``the minimal probability for a successful time-interval" at the beginning of $[k, k+1)$. We shall next prove that for a ``successful time-interval'' agent $s$ moves a distance bounded away from zero by a constant.

To find a bound on the minimal distance that agent $s$ will surely travel if a successful heading was selected (with probability greater than $\frac{1}{2n}$), we must compute the smallest distance that agent $s$ can move ahead unimpeded by any other agent entering its sensing area.

Assume the heading of an agent $s$ at the sharpest corner of the convex-hull is inside its associated $\beta^*_s$ as shown in \cref{fig:UL_SharpestAngle}. Agent $s+1$ (adjacent to $s$ on the convex-hull) and agent $s$ define a side of the convex-hull, therefore $p_{s+1}$ is located somewhere along this side of $\alpha_s$.

By assumption, the agents of the system reside in a convex region contained in the front half-plane of the moving agent $s$. The geometry of the convex region may change in time as other agents can possibly move. However, due to the bounded speed of all agents, the region where all agents reside during some time $d t\le\Delta t$ after the motion of agent $s$ started will be contained within a dilation of the original convex-hull by a disc of radius $d t$.

Thus agent $s$ will keep advancing at least until its forward moving sensing region intersects the dilation of the convex-hull of the agents with the same speed (or does not sense any other agents during this time-interval). Since agents $p_{s-1}$, $p_s$, and $p_{s+1}$ are successive agents on the convex hull, the convex hull at the beginning of the time-step resides completely within the $\alpha_s$ angle: the smaller region limited by the half-lines $[p_s p_{s+1})$ and $[p_s p_{s-1})$. Therefore the dilation of the convex hull is included within the dilation of these half-lines. To lower bound the step made by agent $s$, we study when the dilation of these half-lines first intersect the translation of the sensing region of agent $s$.

The dilation of the half-lines consists in the union of three simple geometric objects: the translation of the two half-lines $[p_s p_{s\pm1})$ in the outward direction orthogonal to them, and a circular arc with fixed center $s$ and angle $\pi - \alpha_s$ and with growing radius.

Assume the translation of the sensing region of agent $s$ first intersects the dilation of the half-lines on one of their translations. In this case, from simple considerations, the first intersection happens on the translation of the half-line $[p_s p_{s\pm1})$ with smallest angle with respect to the original half-plane of sensing of agent $s$. Without loss of generality assume this happens on the $[p_sp_{s+1})$ side.

\begin{figure}[tbhp]
  \centering
    \includegraphics[width=0.6\textwidth]{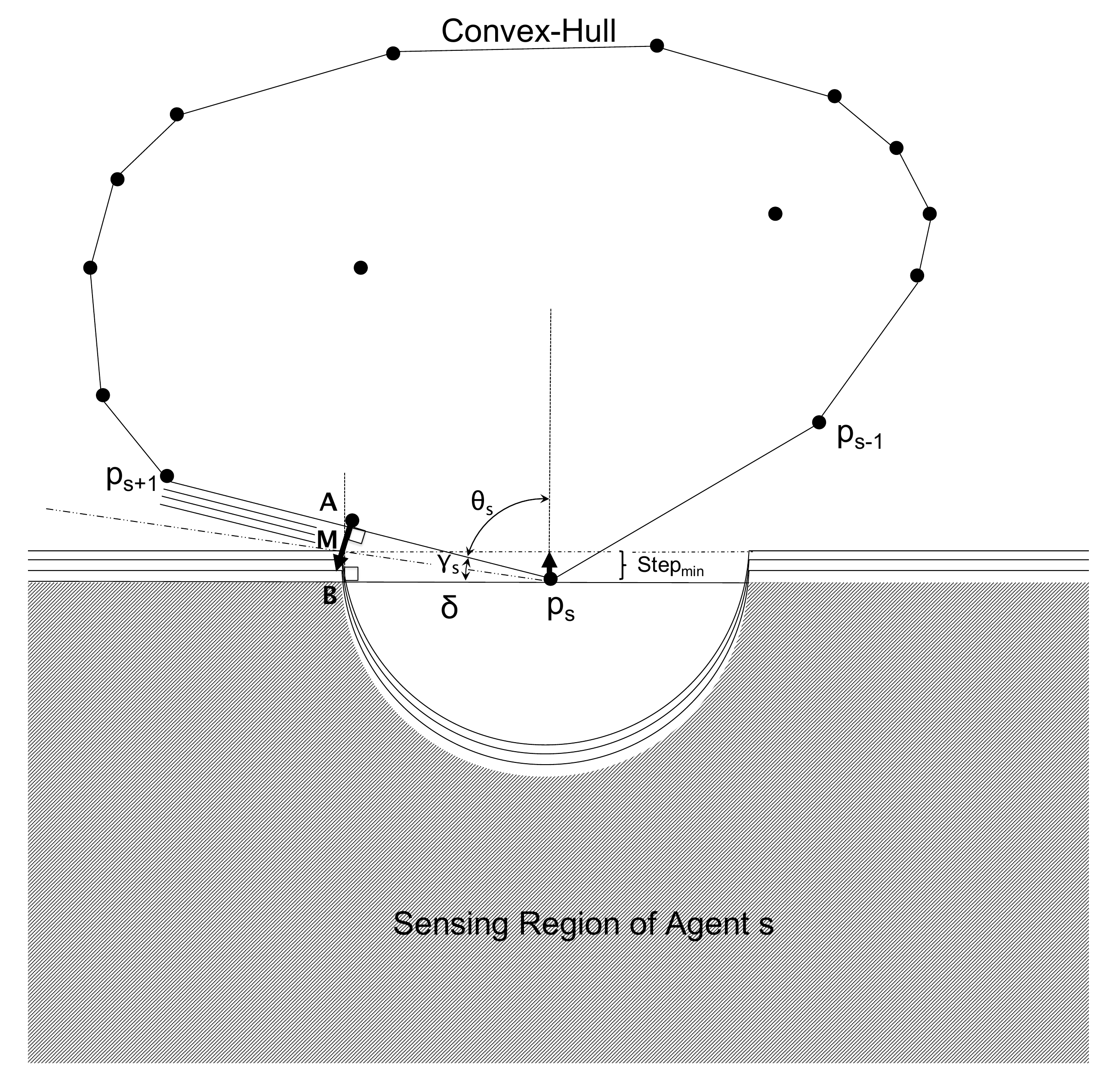}
    \caption{The line through $p_s$ and $M$ is the bisector of $\gamma_s=\frac{\pi}{2}-\theta_s$, and $\|MB\|=\|AM\|$ is the minimal displacement of agent $s$ before being possibly stopped by the constellation (or running out of time).}
      \label{UL_StepMin1}
\end{figure}

We then clearly see that the intersection between the sensing region of agent $s$ and the translating half-line will happen when this half-line intersects the corner point at distance $\delta$ of agent $s$ on the $[p_sp_{s+1})$ side. This happens at point $M$ which is located a distance of $\|BM\|=\delta\tan{\frac{\gamma_s}{2}}$, where $\gamma_s$ is $\frac{\pi}{2}-\theta_s$ (see \cref{UL_StepMin1}). Since the heading of agent $s$ is assumed to be within the $\beta_s^*$ region, we have
$$ 
 \theta_s\le\frac{1}{2}(\alpha_s+\beta^*_s)
 =\frac{1}{2}(\alpha_s+\frac{1}{2}(\pi-\alpha_s))=\frac{1}{4}(\alpha_s+\pi).
 $$
We have by \cref{UL_SharpestAngleValue} that $\alpha_s \leq \pi(1-\frac{2}{n})$, therefore
\begin{equation}\label{Theta_s}
\theta_s \leq \frac{1}{4}(\pi(1-\frac{2}{n})+\pi)=\frac{\pi}{2}(1-\frac{1}{n}).
\end{equation}
And since $\gamma_s\overset{\Delta}{=}\frac{\pi}{2}-\theta_s$, we have that
$
\gamma_s \geq \frac{\pi}{2}-\frac{\pi}{2}(1-\frac{1}{n})=\frac{\pi}{2n}
$,
and therefore (as $0<\frac{\gamma_s}{2}<\frac{\pi}{2}$) we have that
 $$
 \|BM\|\ge \delta \tan{\frac{\pi}{4n}}\ .
 $$
 This shows that agent $s$, in this case, will move by at least a distance of
$$ Step_s \ge \delta \tan{\frac{\pi}{4n}} $$
before possibly being stopped by the motion law we defined. 

Assume now that the translation of the sensing region of agent $s$ first intersects the dilation of the half-lines $[p_sp_{s\pm 1})$ on the dilating open circular arc around $s$.


Since the heading of agent $s$ is inside the $\beta_s^*$ angle, then agent $s$ moves ``forward'', in the sense that if the velocity vector defines an axis, then the projection of the displacement of agent $s$ onto that axis is positive. Assume that the first intersection between the constant angle growing radius circular arc and the translated half-circle border of agent $s$'s sensing region is on a corner point of the border of the sensing region of $s$ at distance $\delta$ from it. As the normal to the circular arc of the sensing region points towards the current position of agent $s$ and the normal to the growing radius circular arc points towards the initial position of agent $s$, we then have that the projection of the vector $p_s(t+dt)-p_s(t)$ is negative onto the axis of displacement of $p_s$, which is a contradiction. Furthermore, we also have that any intersection with the half-line borders of the sensing region of agent $s$ is clearly not the first possible one. Therefore the first intersection of the growing radius circular arc and of the translating border of sensing of agent $s$ necessarily happens in the interior of the half-circle border of the sensing region of agent $s$.


\begin{figure}[tbhp]
  \centering
    \includegraphics[width=0.55\textwidth]{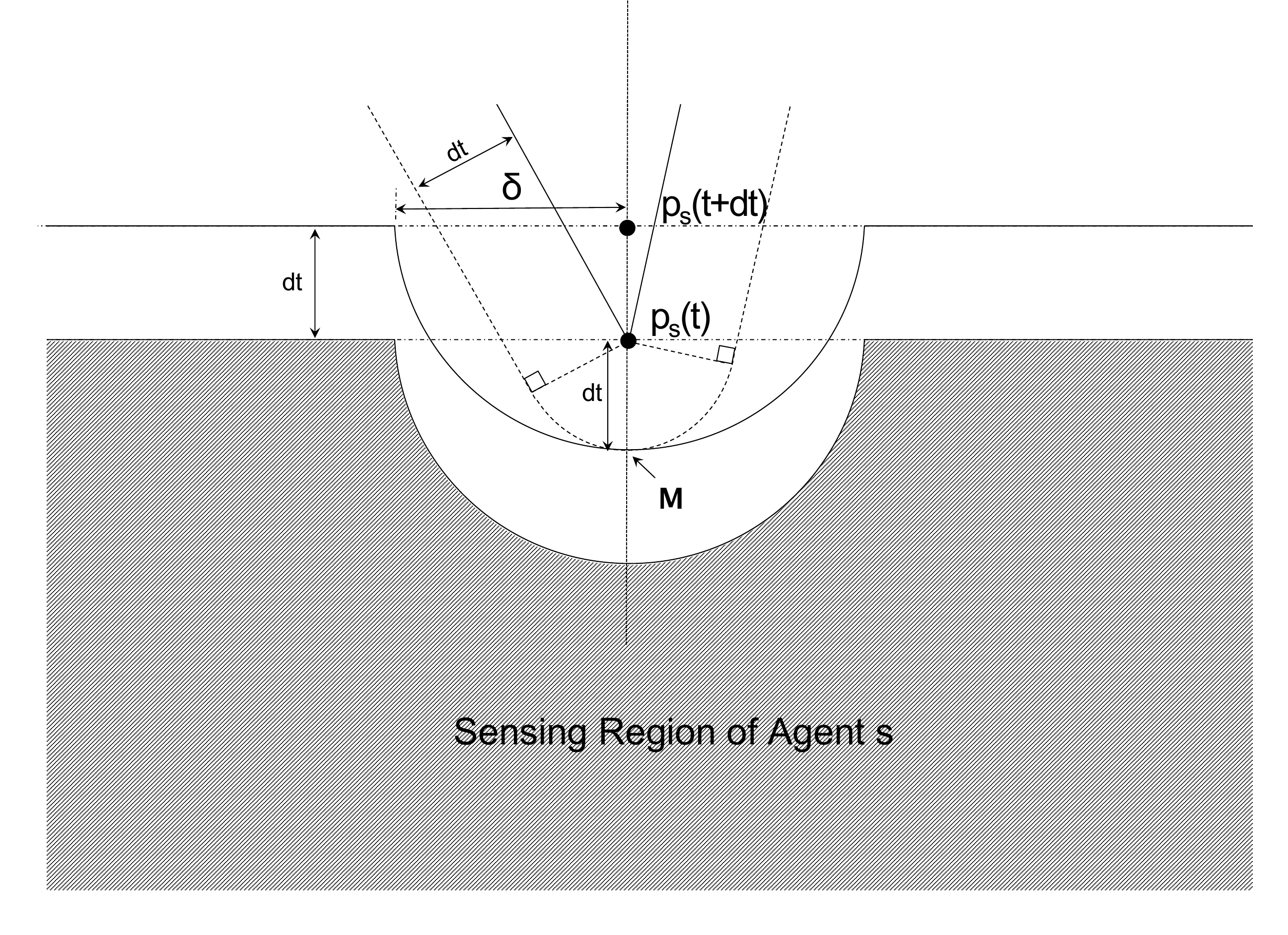}
    \caption{The circular arcs are tangent implying that their centers and the intersection point $M$ are aligned. As ${\|p_s(t)p_s(t+dt)\| = \|p_s(t)M\| = dt}$ and ${\|p_s(t+dt)M\| = \delta}$, necessarily ${dt = \frac{\delta}{2}}$ is the minimum displacement of agent $s$ before being possibly stopped by the constellation (or running out of time).}
      \label{UL_StepMinSharp}
\end{figure}

We thus have that, in this case, the first intersection would happen in the interior of both circular arcs. This further implies, since we are considering the first intersection, it happens when both circles are tangent, as shown in \cref{UL_StepMinSharp}. Therefore the center of each circle and the intersection point, denoted by $M$, are aligned. The center of the dilating circle, of radius $dt$, is the original position $p_s(t)$ of agent $s$. The center of the translating half-circle border, of radius $\delta$, of the sensing region is the new position $p_s(t+dt)$ of agent $s$ at intersection time, separated by $dt$ from the original position. Therefore the time of first intersection is $t+dt$ with $dt = \frac{\delta}{2}$.

This shows that agent $s$, in this second case, will move by at least a distance of
$$ Step_s \ge \frac{\delta}{2} $$
before possibly being stopped by the motion law we defined. However
$$ \frac{\delta}{2} \ge \delta\tan{\frac{\pi}{4n}} \iff n\ge \frac{\pi}{4\arctan{\frac{1}{2}}} \approx 1.7$$
which is always true since we have $n\ge 2$ agents.

We therefore conclude that for all cases, if the heading of agent $s$ falls inside $\beta_s^*$, then it moves by at least
$$ Step_s \ge \delta \tan{\frac{\pi}{4n}} $$
before possibly being stopped by the motion law we defined.




The minimal displacement of the agent defining the sharpest corner of the convex-hull during one time-interval is bounded by the smallest value between $Step_s$ above and the physical limit due to the travel speed $v=1$, i.e. $v\Delta t = 1$, hence we have that the bound on the step of the agent at the sharpest corner of the convex-hull is 
\begin{equation}\label{GeometricStep}
Step_{min} = \min\{\delta \tan{\frac{\pi}{4n}};1\}.
\end{equation}

\end{proof}

\subsubsection{The Lyapunov Function}\

A function is called Lyapunov if it maps the state of the system to a non negative value in such a way that the system dynamics causes a monotonic decrease of this value. If the Lyapunov function reaches zero only at desirable states of the system and we prove that the dynamics leads the Lyapunov function to zero, we can argue that the system converges to a desirable state.

For the proof of system convergence, let us define variables $l_{ij}(t)$ as follows:
\begin{equation}\label{Lyapunov_lij}
l_{ij}(t)= \left\{\begin{alignedat}{2}
    & 0, &&\quad 0\leq d_{ij}(t) \leq \delta\\
    & d_{ij}(t), &&\quad d_{ij}(t)\geq \delta \\
  \end{alignedat}\right.
\end{equation}
and a global variable $c(t)$:
\begin{equation}\label{Lyapunov_c}
c(t)= \left\{\begin{alignedat}{2}
    & 0, &&\quad \exists p_c(t)\in \mathbb{R}^2 \:\: \forall i \;:\; \|p_i(t)-p_c(t)\| < \delta\\
    & 1, &&\quad otherwise \\
  \end{alignedat}\right.
\end{equation}
so that if $c(t)=0$ we have that the system is confined in a disc of radius $\delta$ in the plane.

Let us define the following Lyapunov function:
\begin{equation}\label{lyapunov}
\mathcal{L}(P(t))=c(t)\sum_{i=1}^{n}\sum_{j=1}^{n} l_{ij}(t) .
\end{equation}


By \cref{dijNeverIncrease,NeverLoseFriendCorollary} we have that $l_{ij}$ can never increase, hence $\mathcal{L}(P(t))$ never increases. We shall next prove that with a probability which is finite and bounded away from zero by a constant, $\mathcal{L}(P(t))$ decreases by a positive and bounded away from zero quantity, until it reaches the value zero (within a finite expected time), which will then give us that $\mathcal{L}(P(t))$ will go to zero in finite expected time. 

\begin{lemma}\label{ShrinkMinLemma}
If at time $t$, the beginning of a time-interval, there is an agent $j$ distant more than $\delta$ from the agent $s$, currently located at the sharpest corner of the convex-hull, the probability that $l_{sj}(t)-l_{sj}(t+1)$ is at least a constant bounded away from zero, is higher than $\frac{1}{8n}$. 
\end{lemma}

Note that by \cref{ShrinkMinLemma}, we have that in finite expected time all agents will necessarily be confined inside a disc of radius $\delta$.

\begin{proof}
In order to evaluate the influence of agent $s$'s motion on the Lyapunov function defined by \cref{lyapunov}, we need to see how a step bigger than or equal to $Step_{min}=\min\{\delta \tan{\frac{\pi}{4n}},1\}$ can influence $l_{ij}(t)$ in the sum defining $\mathcal{L}(P(t))$.

Since agents $s$ and $j$ are assumed to be separated at time $t$, then gathering has not yet been achieved at the beginning of the time-interval. 

Suppose agent $j$ is located somewhere in the region shown in \cref{UL_Shrink1} as region $2$. We can easily lower bound the probability that it will remain stationary during the entire initial motion of agent $s$ as follows. 

\begin{figure}[tbhp]
  \centering
    \includegraphics[width=70mm]{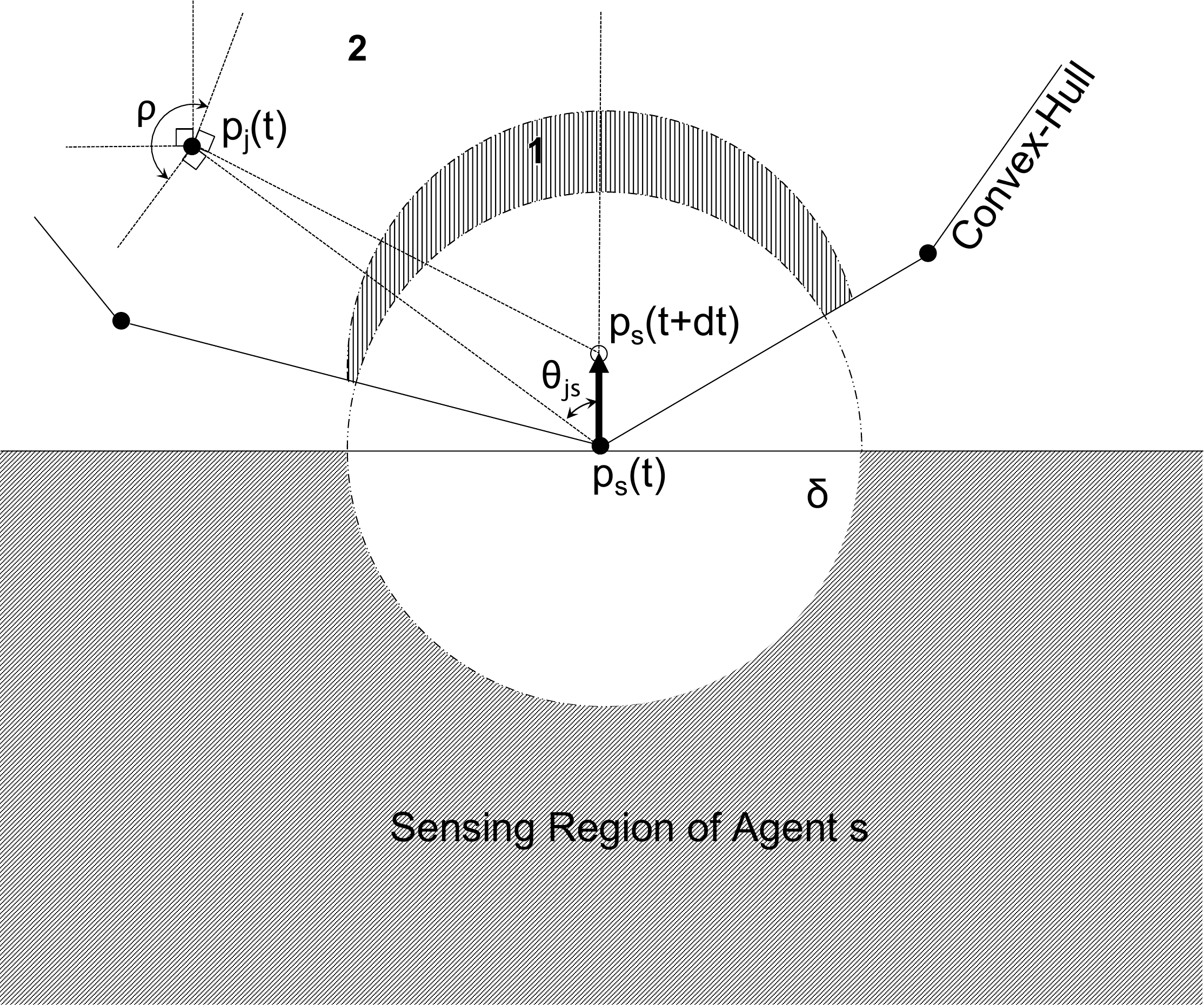}
        \caption{Region $1$ is the area in the convex-hull of the agents' locations at time $t$, swept by moving the $\delta$ radius disc in the heading direction of $s$, from $p_s(t)$ to $p_s(t+dt)$, where $dt\le 1$ defines the time for which agent $s$ moves at the beginning of the time-interval without being stopped, and then excluding the area of the $\delta$ radius disc at $p_s(t)$. Note that $dt\le 1$ as $s$ might be allowed to move again depending on the rest of the constellation after being stopped at time $t+dt$. Region $2$ is the area of the convex-hull which is not included in the disc of radius $\delta$ centered at $p_s(t)$ and also not in region $1$. If agent $j$ is inside region $2$ and its heading is inside the sector of angle $\rho \ge \frac{\pi}{2}$, it stays put while agent $s$ initially travels, at most until $j$ enters the sensing area of agent $s$. 
        }
      \label{UL_Shrink1}
\end{figure}

%

If agent $j$'s heading falls inside the $\rho\ge\frac{\pi}{2}$ angle, as illustrated in \cref{UL_Shrink1}, then it can sense all points in $[p_s(t),p_s(t+dt)]$ and thus will stay put throughout the initial motion of agent $s$ until $s$ stops at time $t+dt \le t+1$. Thus agent $j$ stays put during the entire initial motion of agent $s$ with probability at least $\frac{1}{4}$. 

Now suppose that agent $j$ is located in region $1$ shown in \cref{UL_Shrink1}. We can apply a similar reasoning by replacing $dt$ by $dt'\le dt$ where $t+dt'$ is the first time when ${p_s(t+dt')}$ comes within range $\delta$ of agent $j$'s initial position $p_j(t)$. Likewise we can define an angle $\rho'\ge\frac{\pi}{2}$ ensuring that if agent $j$'s heading falls within $\rho'$ then agent $j$ will sense all points within ${[p_s(t),p_s(t+dt')]}$ which implies to stay put within $[t,t+dt']$, and thus come within distance $\delta$ of agent $s$ during $s$'s initial displacement. This will happen with probability at least $\frac{1}{4}$.

With probability $\frac{1}{4}$ or larger, agent $j$ will not move, either during the entire initial movement of agent $s$ if $j$ is in region 2 or during the initial movement of agent $s$ until it comes within range $\delta$ of this agent if $j$ is in region 1, since, in both cases, agent $j$ will sense agent $s$ during its entire corresponding initial movement. From now on we will assume that this event occurs, i.e. that if $j$ is in region 2 then $j$ stays put during the entire initial movement of $s$, and if $j$ is in region 1 then $j$ stays put until it is within distance $\delta$ from $s.$

As agent $s$ starts to move, depending on whether it is in region 1 or 2, agent $j$ does one of the following: 
\begin{enumerate}[(i)]
\item stays put and becomes, due to the motion of $s$, within range of $\delta$ from $s$ and then $l_{sj}$ and $l_{js}$ will decrease by at least $\delta$ each.
\item remains stationary at a distance bigger than $\delta$ from $s$ for the entire motion of $s$, and in this case we can bound the decrease of $l_{sj}$ and $l_{js}$ as follows:
\end{enumerate}

Let us consider
\begin{equation}\label{Shrink_s_1}
	Shrink_{s,dt} = d_{sj}-\sqrt{d^2_{sj}+Step_s^2-2d_{sj}Step_s\cos{\theta_{sj}}} 
\end{equation}
where (see \cref{UL_Shrink1}) $d_{sj}$ is the mutual distance of agents $s$ and $j$ at the beginning of a time-interval $t$, and $\sqrt{d^2_{sj}+Step_s^2-2d_{sj}Step_s\cos{\theta_{sj}}}$ is their mutual distance at time $t+dt$ when agent $s$ first stops. Using \cref{dijNeverIncrease}, the shrink of distance $Shrink_s$ between agents $s$ and $j$ during the time-interval is at least: 
\begin{equation} 
	\label{eq: shrink_s at least shrink_s dt}
	Shrink_s \ge Shrink_{s,dt}.
\end{equation}

Let us consider $Shrink_{s,dt}$ as the value of a function $Sh$ of three variables $d$, $St$, $\theta$: 
$$
Sh(d, St, \theta) \overset{\Delta}{=} d-\sqrt{d^2+St^2-2dSt\cos\theta}.
$$
We have that 
\begin{equation}\label{dSh_dd}
\begin{split}
\frac{\partial Sh}{\partial d}=1-\frac{d-St\cos \theta}{\sqrt{d^2+St^2-2dSt\cos \theta}}\ge0\\
\forall \; d>0\;; \; St>0\;; \;\theta \in [0, \frac{\pi}{2})
\end{split}
\end{equation}

\begin{equation}\label{dSh_dtheta}
\begin{split}
\frac{\partial Sh}{\partial \theta}=-\frac{dSt\sin \theta}{\sqrt{d^2+St^2-2dSt\cos \theta}}\le0\\
\forall \;  d>0\;; \; St>0\;; \; \theta \in [0, \frac{\pi}{2})
\end{split}
\end{equation}

\begin{equation}\label{dSh_dSt}
\begin{split}
\frac{\partial Sh}{\partial St}=-\frac{St-d\cos \theta}{\sqrt{d^2+St^2-2dSt\cos \theta}}\ge0\\
\forall \; d>0\;; \;St\le d\cos\theta\;; \; \theta \in [0, \frac{\pi}{2})\\
\end{split}
\end{equation}
and
\begin{equation}
	\label{eq: Domain d St theta}
	d\in[\delta,\infty)\,;\,St\in[St_{min}, d\cos{\theta_s})\,;\,\theta\in[0, \theta_s)\subset[0,\frac{\pi}{2}).
\end{equation}
If agent $j$ is in region $2$, as shown in \cref{UL_Shrink1}, then due to \cref{Theta_s,dSh_dd,dSh_dtheta,dSh_dSt,eq: Domain d St theta}, we have 
\begin{equation}\label{ShrinkBound}
Sh \geq Sh (d_{min}, \theta_{max}, St_{min})  \overset{\Delta}{=} Sh_{min}
\end{equation}
where $d_{min} = \delta$, $\theta_{max} = \frac{\pi}{2}(1-\frac{1}{n})$, and $St_{min} = Step_{min}$, thus:
\begin{equation}
	\label{eq: Shrink_bound_general}
	Sh_{min} = \delta-\sqrt{\delta^2+Step_{min}^2-2\delta Step_{min}\sin{\frac{\pi}{2n}}}\ .
\end{equation}

For finite $n\ge 1$ we can show that $Sh_{min}$ is a constant bounded away from zero
\begin{equation}
	\label{eq: Sh_min pos}
	Sh_{min} > 0,
\end{equation}
see \cref{App: Proof of sh_min positive} for full details. Note that:
\begin{equation}
	\label{eq: Sh_min smaller delta}
	Sh_{min} \le \delta.
\end{equation}

We have therefore proven that if agent $j$ is further than $\delta$ away from agent $s$, the agent at the sharpest corner of the convex-hull, at the beginning of the time-interval, and in case of a ``successful time-interval'', i.e. the heading of agent $s$ is within $\beta_s^*$, then, due to \cref{eq: shrink_s at least shrink_s dt}, the Lyapunov function will decrease by at least a strictly positive constant $Shrink_{min}$, which is the minimum of $2\delta$ and $2Sh_{min}$. By \cref{eq: Sh_min pos,eq: Sh_min smaller delta}, we get: 
\begin{equation}
	\label{eq: Shrink_min positive}
	Shrink_{min} = \min\{2\delta,2Sh_{min}\} = 2Sh_{min} >0.
\end{equation}

We shall call such an event, described in the previous paragraph and giving a shrink in the Lyapunov function of at least $Shrink_{min}$, a ``fully successful time-interval'', abbreviated ``$\textit{FSTI}$''. Since the headings of agent $j$ and $s$ are drawn independently, the event ``$j$ stays put during the initial motion of $s$'' happens with probability greater than $\frac{1}{4}$ conditionally to the event ``successful time-interval'', which itself happens with probability greater than $\frac{1}{2n}$ (given in \cref{UL_SharpestAngleToMove}). Therefore the probability of a ``$\textit{FSTI}$'', which is the intersection of both previous events, to occur is lower-bounded by 
\begin{equation}
	\label{eq: Pr(FSTI)}
	Pr(\textit{FSTI}) \ge p \overset{\Delta}{=}  \frac{1}{8n}. 
\end{equation}

\end{proof}

To complete the proof of \cref{Theorem}, the Lyapunov function $\mathcal{L}(P(t))$ will become zero when there is a point in $\mathbb{R}^2$ whose distance to all agents is smaller then $\delta$. We study the expected number of time-intervals necessary for this to happen.

Recall that if $Q$ is an event that occurs in a trial with probability $q$, we have that the mathematical expectation $E$ of the number of trials $k_Q$ for the first occurrence of $Q$ in a sequence of independent trials is
\begin{equation}
	\label{eq: exp k_Q}
	E[k_Q]=q+2(1-q)q+3(1-q)^2q+\hdots=\sum_{k=1}^{\infty}k(1-q)^{k-1}q=\frac{1}{q}.
\end{equation}

In our case, the probability to have a ``fully successful time-interval'' is greater than $p$ as given in \cref{eq: Pr(FSTI)} in all trials, and this event is independent from the past intervals. Thus, using \cref{eq: exp k_Q}, the expectation of the first occurrence of a ``fully successful time-interval'' is finite and upper bounded as follows:
\begin{equation}
	E[k_{\textit{FSTI}}] \le \frac{1}{p} = 8n. 
\end{equation}

Once we have reached a ``fully successful time-interval'', having the next ``$\textit{FSTI}$'' in future intervals is independent from the first one, and we can apply the same reasoning for the second occurrence of this event. By reiterating this argument, we thus get that the expectation of the number of time-intervals $t_{cv}$ for convergence inside a disc of radius $\delta$ is finite and is upper-bounded by
\begin{equation}
	\label{eq: E time conv general}
	E(t_{cv}) \le \ceil[\bigg]{\frac{\mathcal{L}(P(0))}{Shrink_{min}}} 8n 
\end{equation}
since clearly the Lyapunov function equals zero is equivalent to all agents are confined in a region of radius $\delta$. Since the initial value of the Lyapunov function is less than $n(n-1)d_{max}(0)$ (as in the chosen Lyapunov each ``edge" is counted twice), and due to \cref{ShrinkBound,eq: E time conv general}, the expected number of time-intervals to convergence of the system is upper-bounded as follows
\begin{equation}
\label{eq: E time conv general explicit}
	E(t_{cv})\leq 8n\ceil[\Bigg]{\frac{n(n-1)d_{max}(0)}{2\delta\left[ 1-\sqrt{1+(\frac{Step_{min}}{\delta})^2-2 (\frac{Step_{min}}{\delta})\sin{\frac{\pi}{2n}}}\right] }} 
\end{equation}
which is finite, and dependent on the initial constellation, the number of agents $n$, and the radius of the blind-zone $\delta$. We have, following \cref{GeometricStep}, that $\frac{Step_{min}}{\delta} = \min\{\tan{\frac{\pi}{4n}},\frac{1}{\delta}\}$. In particular, for $\delta$ sufficiently small, i.e. $\delta\le\frac{1}{\tan{\frac{\pi}{4n}}}$, that:
\begin{equation}
\label{eq: E time conv small delta explicit}
\begin{split}
	&E(t_{cv})\leq 8n\ceil[\Bigg]{\frac{n(n-1)d_{max}(0)}{2\delta\left[ 1-\sqrt{1+(\tan{\frac{\pi}{4n}})^2-2 (\tan{\frac{\pi}{4n}})\sin{\frac{\pi}{2n}}}\right] }}. 
\end{split}
\end{equation}
Note how the upper-bound is approximately proportional to $\frac{1}{\delta}$. This suggests that the blind-zone could be necessary in order to ensure finite expected time convergence.

\end{proof}

\subsection{Simulation Results and System Behaviour}

In order to approximate continuous sensing in numerical simulations, we further divided each unit time-step $\Delta t$ into smaller time-steps $dt$. In each of these smaller time intervals, the agents' dynamics are similar to those in the discrete case defined in \cref{eq:Dynamics}, except that the agents keep the same heading and jump only a distance of $dt$. Every $\frac{1}{dt}$ small time-steps, we have reached the end of the unit time-step and all agents randomly change their heading. In order to remain within an approximation of the continuous setting, we need to choose $dt$ such that $dt < \delta$, therefore various choices are made for different sizes of blind-zone radii.

Up to scale, the dynamics of the agents are clearly invariant: given a sequence of random orientations, if the edge of the square domain size is multiplied by $\lambda$, and if we multiply the blind-zone radius $\delta$ of agents by $\lambda$, then the dynamics of the agents are exactly a $\lambda$-scaled version of the dynamics they would have been with the same randomised sequence of orientation in the unit square with blind-zone radius $\delta$. We therefore choose to set the domain to a fixed size in all our simulations: agents are initially randomly and uniformly placed in a unit square. 

Given this choice, our domain of interest for the size of the blind-zone is $\delta\le 1$. Since $\tan{\frac{\pi}{4n}} \le 1$ for $n\ge1$ agents, then such choices of $\delta$ imply $\frac{1}{\delta}\ge 1 \ge \tan{\frac{\pi}{4n}}$ therefore the expected time to convergence is given by \cref{eq: E time conv small delta explicit} in our simulations.

Furthermore, in order to avoid undesired bias due to the initial distribution of the points, agents initial positions are randomly and independently resampled for each run. Thus the initial Lyapunov value or likewise the initial maximum distance between any two agents is a random value. However, we can bound these values given our choice of fixed domain size. For instance, thanks to Pythagoras, $d_{max}(0)\le \sqrt{2}$ and thus \cref{eq: E time conv small delta explicit} becomes
\begin{equation}
	\label{eq: E time conv small delta explicit fixed square experiments}
	\begin{split}
		&E(t_{cv})\leq 8n\ceil[\Bigg]{\frac{\sqrt{2}n(n-1)}{2\delta\left[ 1-\sqrt{1+(\tan{\frac{\pi}{4n}})^2-2 (\tan{\frac{\pi}{4n}})\sin{\frac{\pi}{2n}}}\right] }},
	\end{split}
\end{equation}
which we will use as a reference bound when considering several runs. Note that here we have made an abuse of notation. Indeed, in \cref{eq: E time conv small delta explicit}, the time for convergence $t_{cv}$ implicitly depends on the initial configuration of the agents and so the expected value is implicitly conditioned with respect to it. To numerically approximate it, we should simulate many runs starting from the same initial distribution. However, in practice, we desire to have a quantity that would only depend on the parameters of the initial swarm (such as their number, spread, and blind-zone size), rather than a quantity biased by the initial distribution. This is why we allow ourselves to resample the distribution of the agents at each run, and the empirical average convergence time we get is actually an estimate of this quantity depending only on the initial parameters of the swarm rather than on its explicit distribution.

Typical simulation results of gathering are shown in \cref{fig:continuous evolution}. As expected, the systems converge to a disk of radius $\delta$, whose centroid wanders in the plane.

\begin{figure}[tbhp]
    \centering
    \subfloat[$n=5$, $\delta=0.02$]{\label{fig: Continuous_memoryless_n_5_delta_0.02}\includegraphics[width=0.45\textwidth]{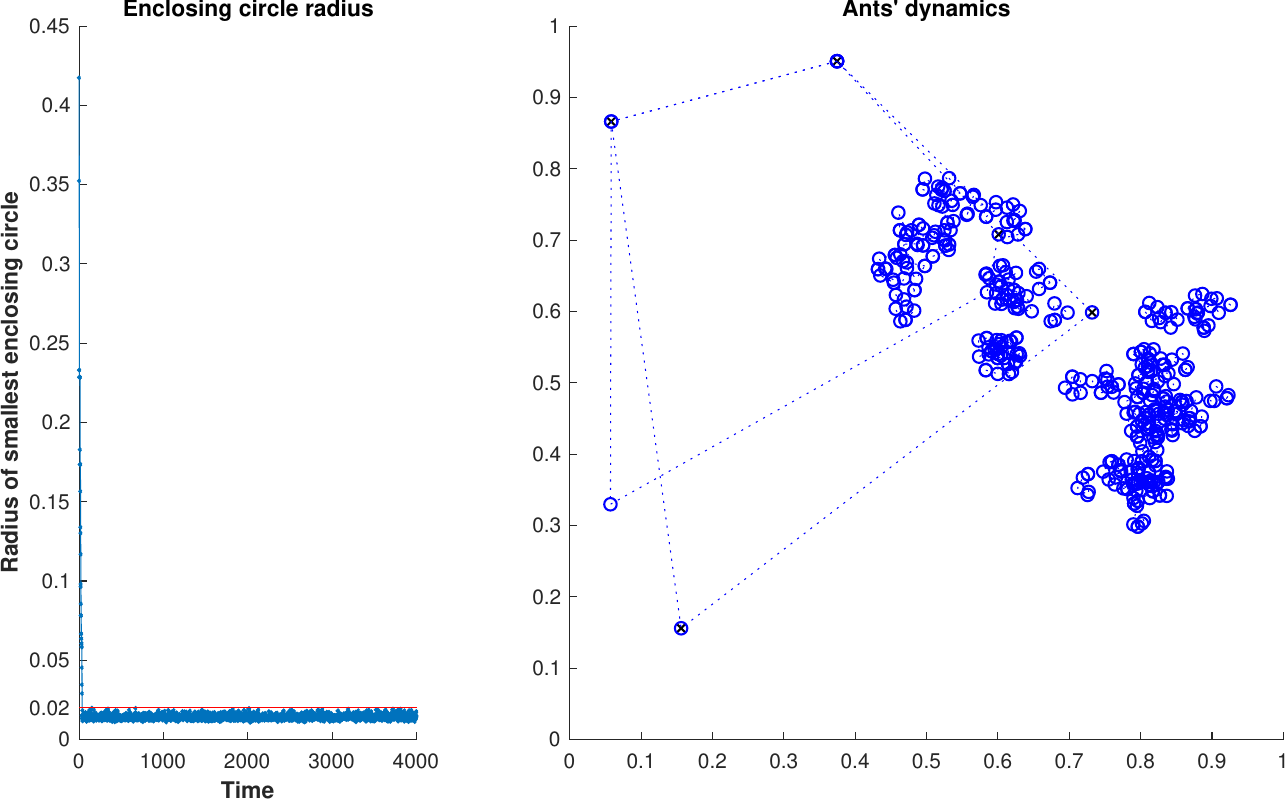}}
    \subfloat[$n=10$, $\delta=0.02$]{\label{fig: Continuous_memoryless_n_10_delta_0.02}\includegraphics[width=0.45\textwidth]{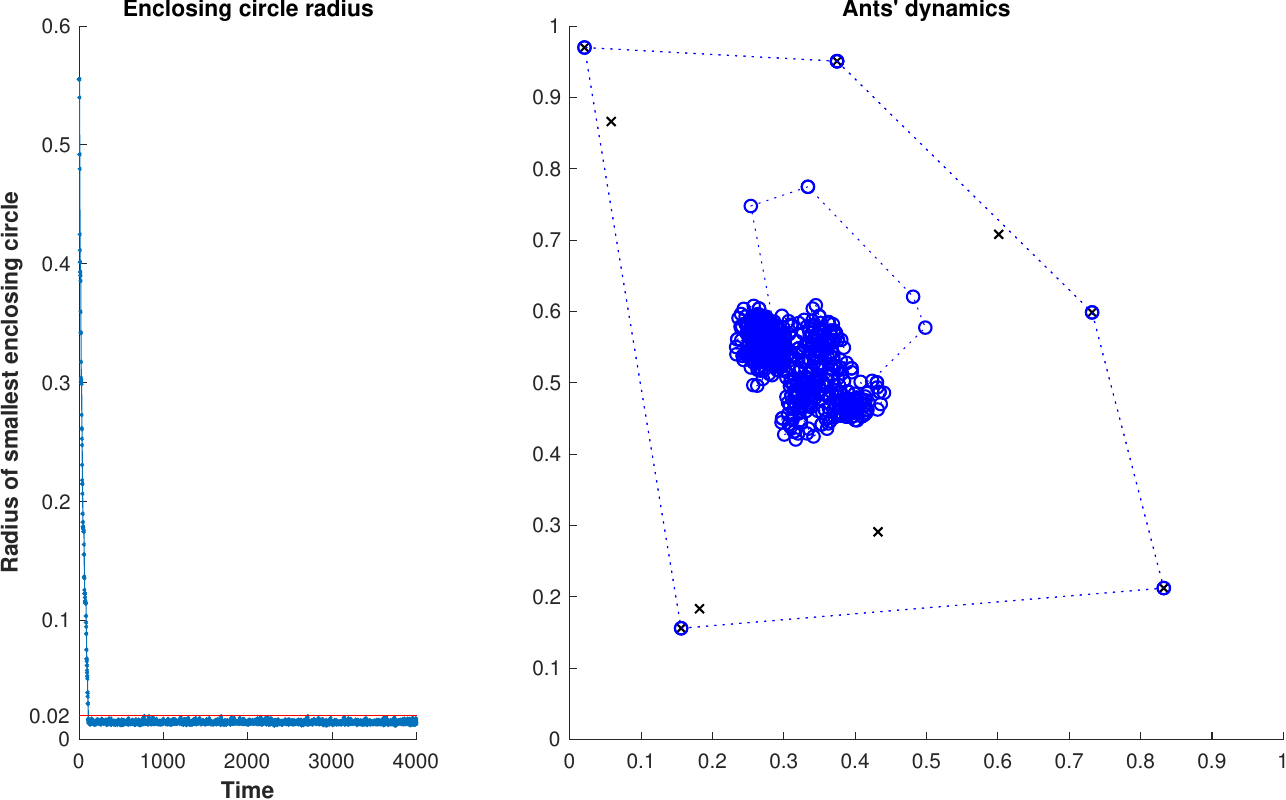}}\\
    \subfloat[$n=50$, $\delta=0.02$]{\label{fig: Continuous_memoryless_n_50_delta_0.02}\includegraphics[width=0.45\textwidth]{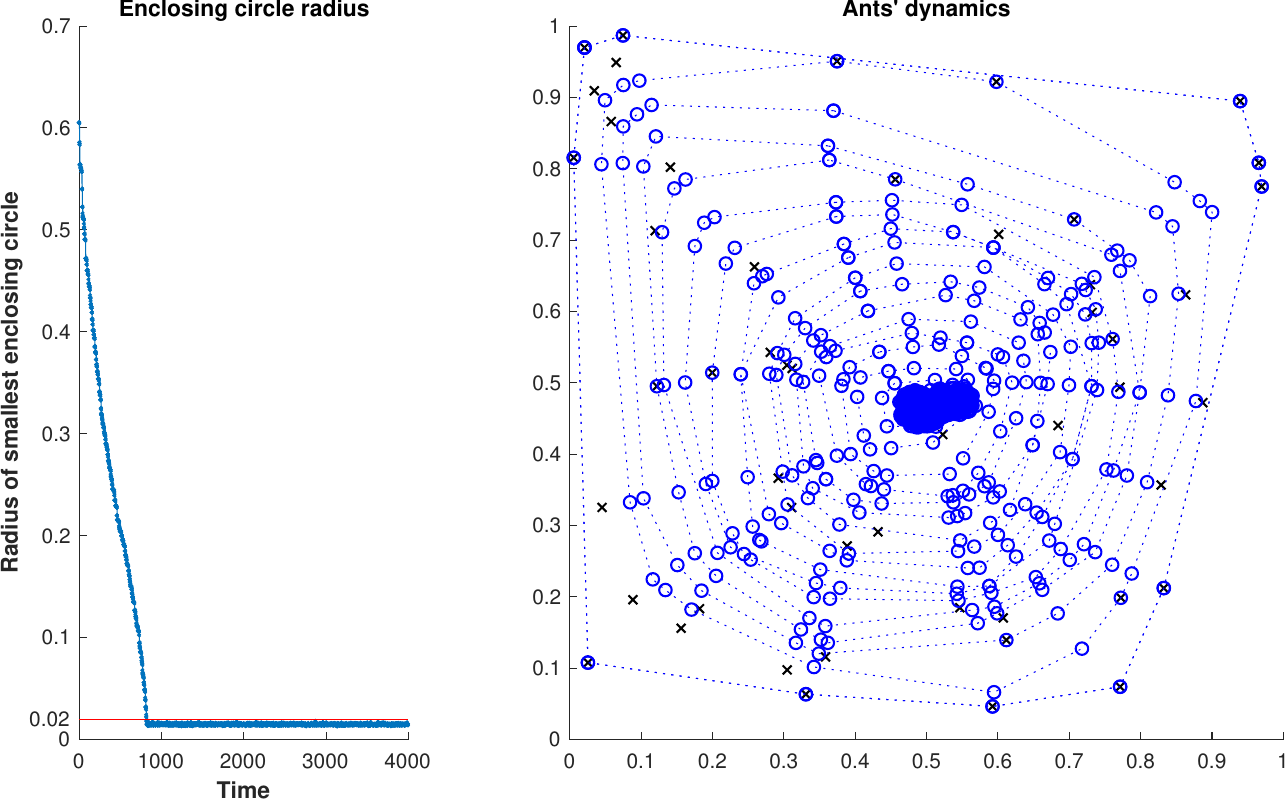}}
    \subfloat[$n=200$, $\delta=0.02$]{\label{fig: Continuous_memoryless_n_200_delta_0.02}\includegraphics[width=0.45\textwidth]{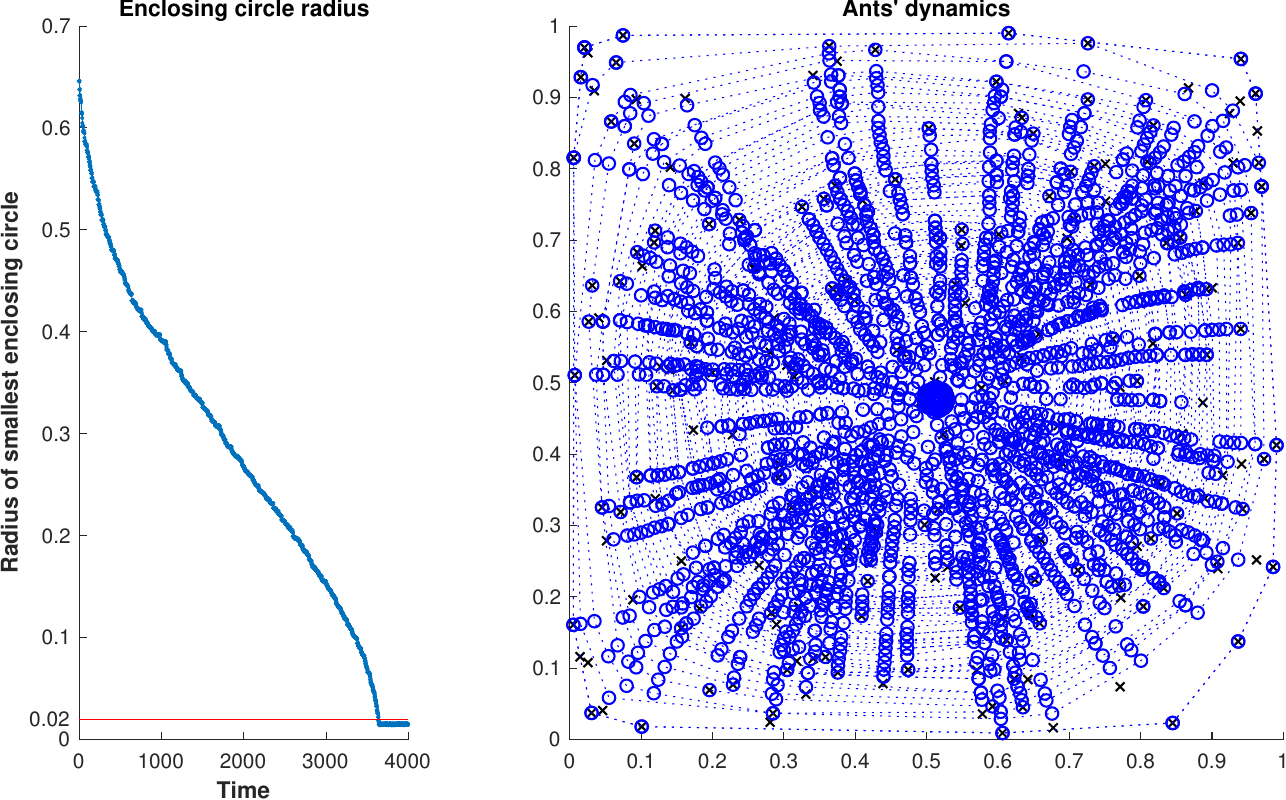}}\\
    \subfloat[$n=10$, $\delta=0.005$]{\label{fig: Continuous_memoryless_n_10_delta_0.005}\includegraphics[width=0.45\textwidth]{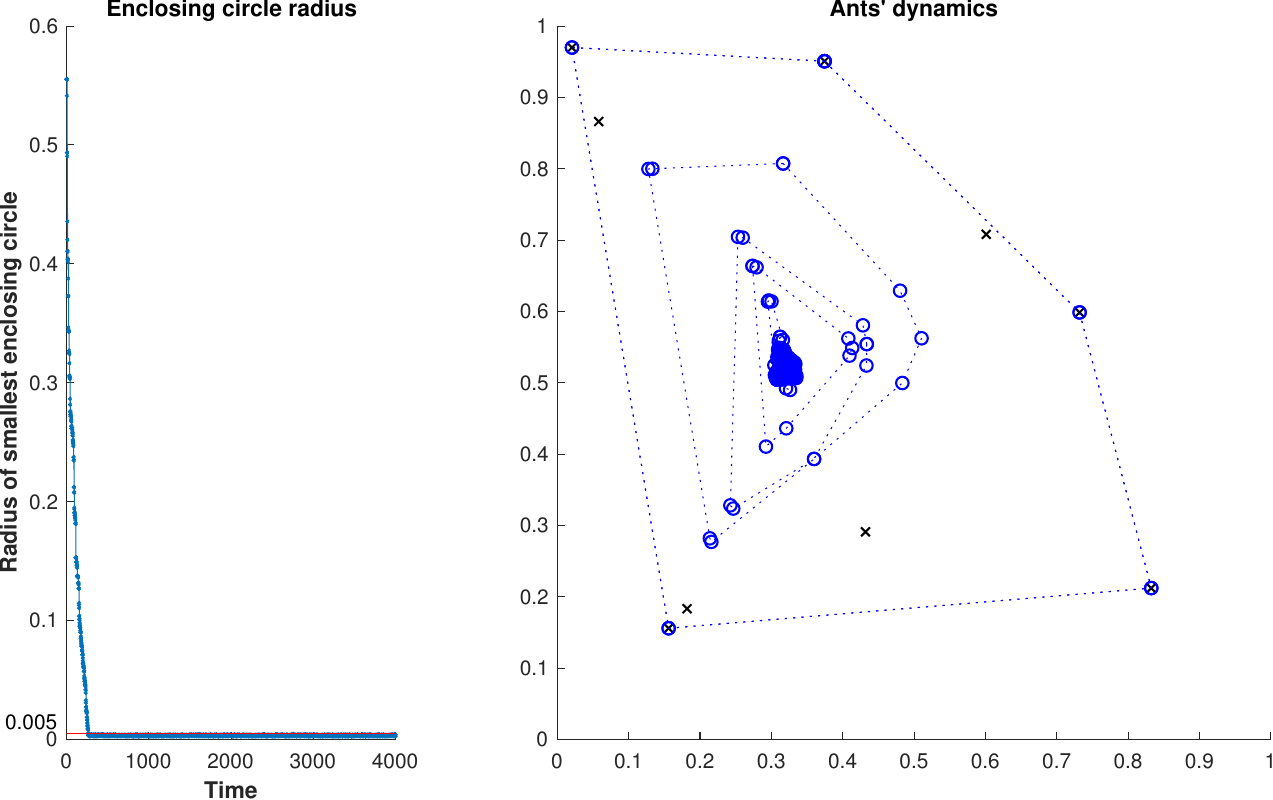}}
    \subfloat[$n=10$, $\delta=0.001$]{\label{fig: Continuous_memoryless_n_10_delta_0.001}\includegraphics[width=0.45\textwidth]{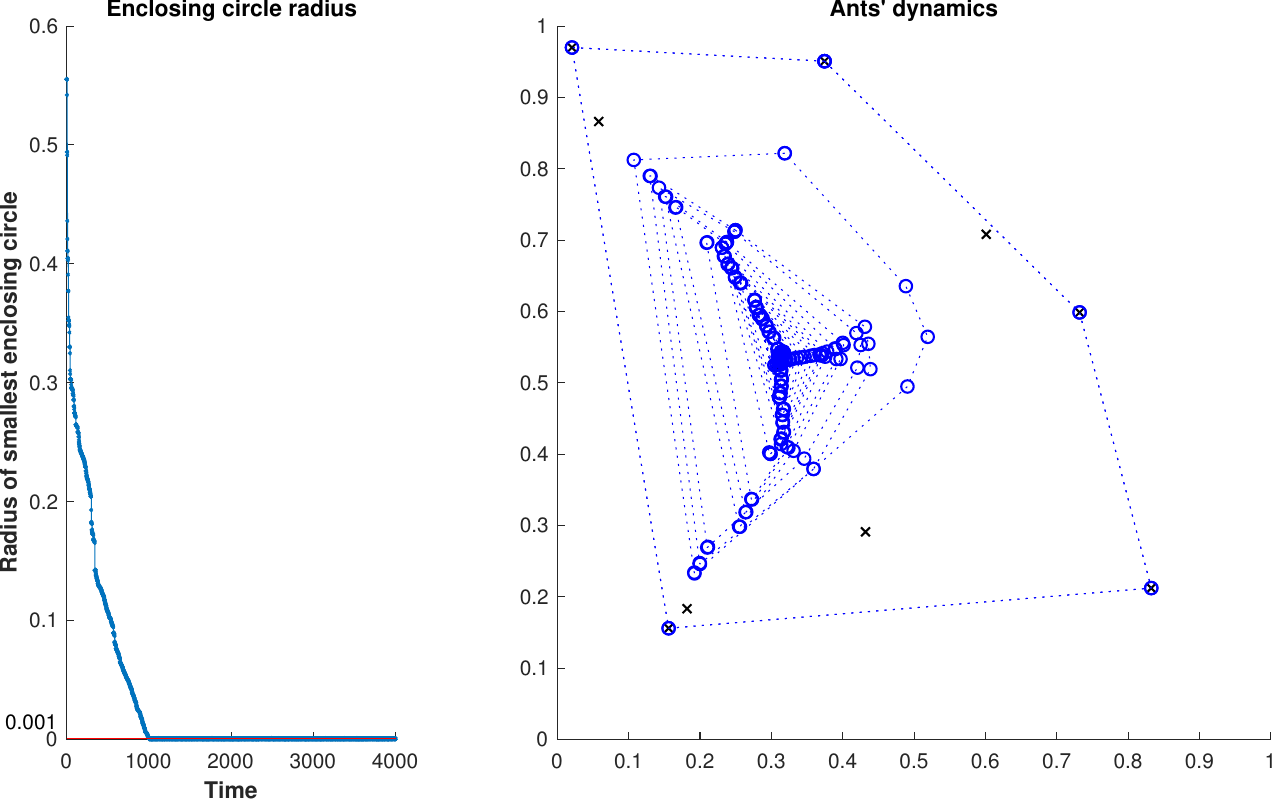}}
    \caption{Simulation results on various number of agents $n$ for the continuous dynamics algorithm and several blind-zone radii $\delta$, with initial random spread in a $1$ by $1$ area. The smaller step-size $dt$ is $0.01$ (respectively $0.002$ and $0.0005$) for $\delta$ equal to $0.02$ (respectively $0.005$ and $0.001$). On the right part of each figure,
    the convex-hull of the system and its associated agents are printed every $50$ time-steps (marked with dashed lines and $\circ$). The initial position of the agents is marked with $\times$. The curve on the left part of each figure shows the decrease in the radius of the smallest enclosing circle of the agents constellation.}
    \label{fig:continuous evolution}
\end{figure}

We found that, for estimating the expected time of convergence, a choice of $1{,}000$ runs provides a reliable estimator. A detailed analysis of this choice is presented as supplementary material in \cref{an subsec: Continuous analysis nb runs}.





\Cref{fig:continuous delta analysis curves} summarize $1{,}000$ simulation results with different blind-zone sizes $\delta$, using $n=10$ agents spread uniformly on the same initial area. Notice how the convergence time is proportional to $\frac{1}{\delta}$, similarly to what we found in the bound \cref{eq: E time conv small delta explicit}. As expected, the average convergence time is below this theoretical bound. However, these results show how pessimistic this worst-case bound is, since it is systematically several orders of magnitude higher than the empirical convergence time. Indeed, the slope of the weighted least-squares interpolation of the results is $1.30$ unit time-steps per invert radius distance of the blind-zone, compared to $5.50\times 10^{5}$ with the same units for the theoretical bound.

\begin{figure}[tbhp]
    \centering
    \subfloat[]{\label{fig: Continuous_convergenceTime_VS_delta_res}\includegraphics[width=0.6\textwidth]{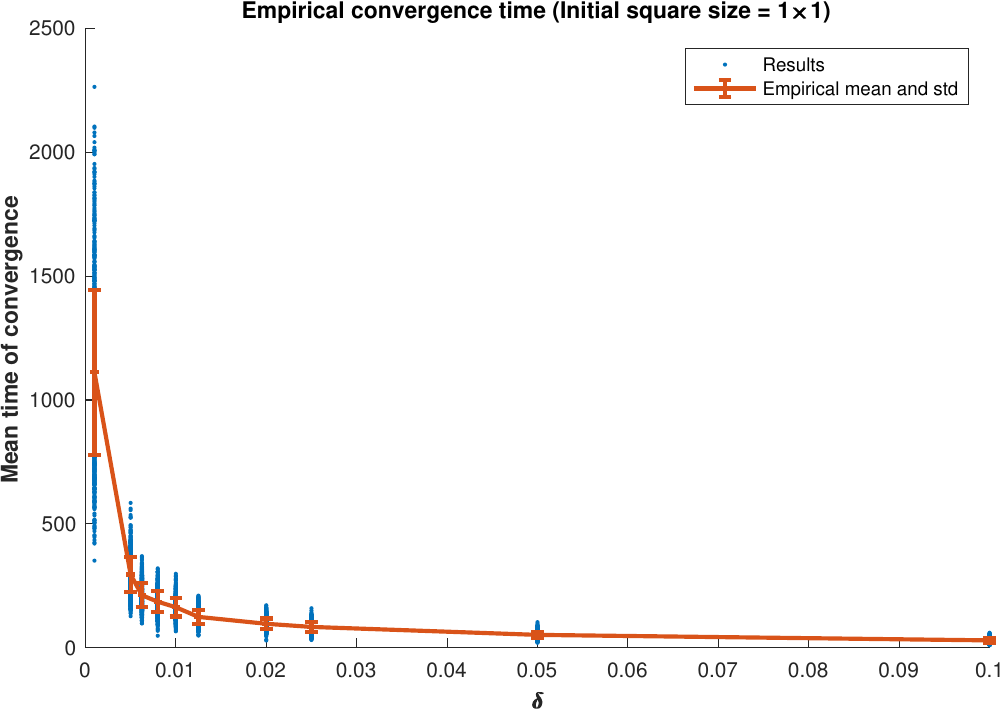}}\\
    \subfloat[]{\label{fig: Continuous_convergenceTime_VS_delta_res_bound}\includegraphics[width=0.6\textwidth]{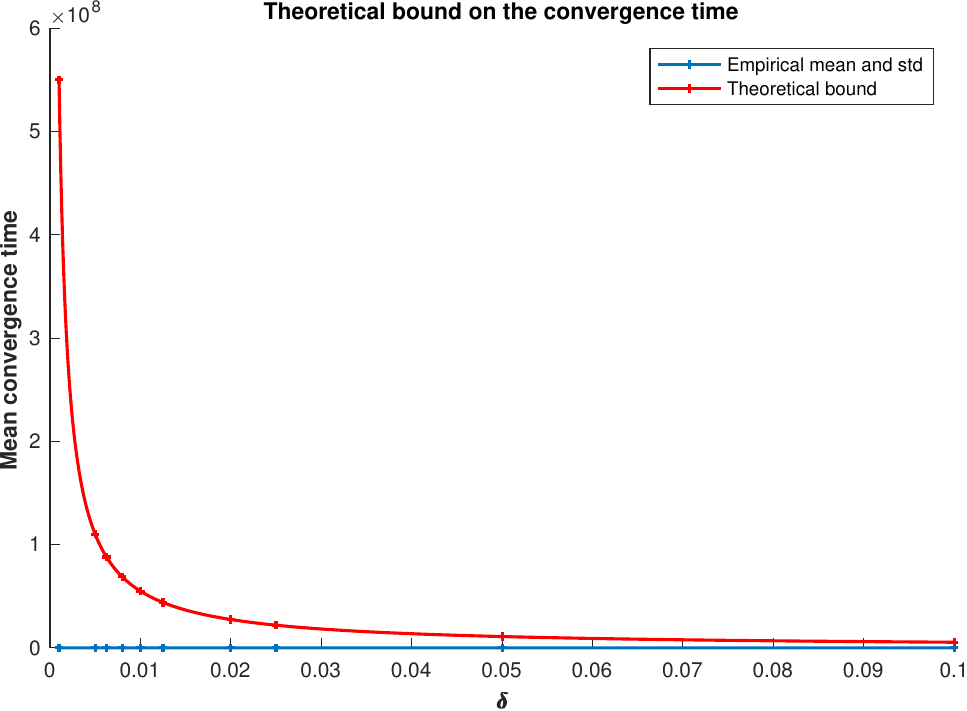}}\\
    \subfloat[]{\label{fig: Continuous_convergenceTime_VS_inv_delta}\includegraphics[width=0.6\textwidth]{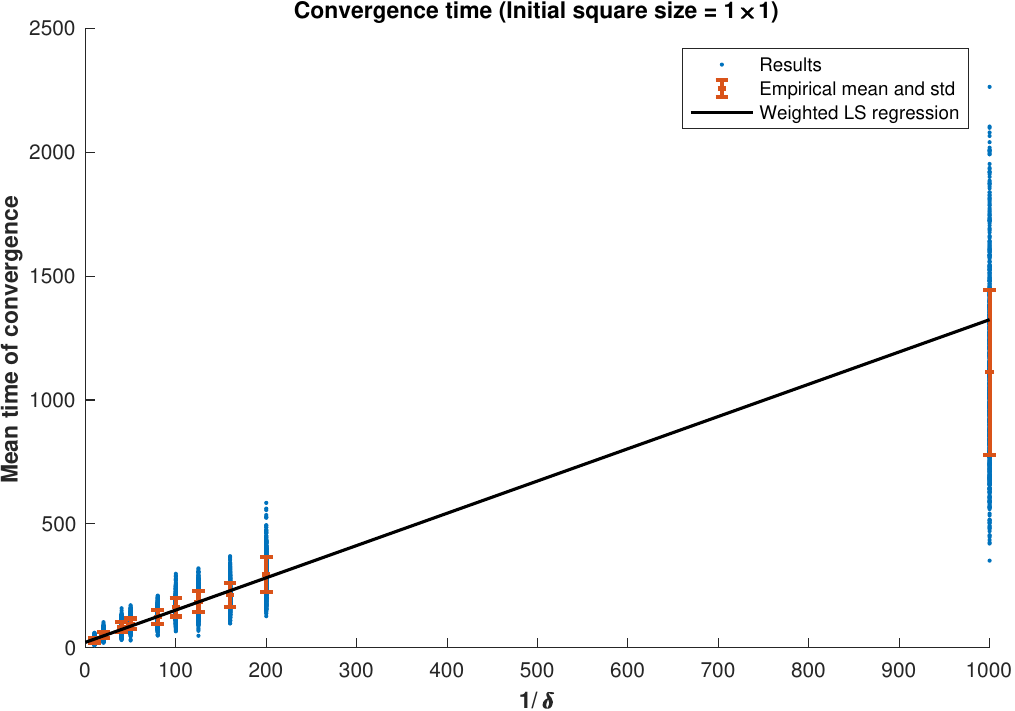}}
    \caption{Analysis of the influence of the radius $\delta$ of the blind-zone on the convergence time. All simulations were set to an initial random spread of $10$ agents in a $1$ by $1$ area. Convergence time was taken for the first time when all agents were gathered in a circle of radius $\delta$. This process ran for $1{,}000$ repetitions with different random initial constellations. \Cref{fig: Continuous_convergenceTime_VS_delta_res}: convergence time vs. radius $\delta$. \Cref{fig: Continuous_convergenceTime_VS_delta_res_bound}: convergence time vs. radius $\delta$ and with the bound given in \cref{eq: E time conv small delta explicit}. \Cref{fig: Continuous_convergenceTime_VS_inv_delta}: convergence  time vs. inverse of the radius $\delta$. On the results in \cref{fig: Continuous_convergenceTime_VS_delta_res,fig: Continuous_convergenceTime_VS_inv_delta}, we superimpose the empirical mean and standard deviation of the results. In \Cref{fig: Continuous_convergenceTime_VS_inv_delta}, the linear graph was obtained using weighted linear least squares fitting to the average results, with weights equal to the squared inverse of the empirical standard deviation of the data to take into account the non uniform variance of the results.}
    \label{fig:continuous delta analysis curves}
\end{figure}



\Cref{fig:continuous n analysis curves} summarize $1{,}000$ simulations with different number of agents, spread uniformly over the same initial area and with same blind-zone radius $\delta = 0.02$ and small time-step $dt = 0.01$. Empirically we find that the effect of the number of agents on the convergence time of the system is linear, similarly to the discrete case. However, in \cref{eq: E time conv small delta explicit fixed square experiments}, the bound is $O(n^3)$, which once more reveals how pessimistic our worst case is. The slope of the weighted least-squares interpolation of the results is approximately $17$ unit time-steps per number of agents, which remarkably is not only within the same order of magnitude as the slope in the discrete version, but also smaller. This suggests that the continuous convergence scheme leads faster to convergence than the discrete rules of motion. This could be due to the fact that in the continuous version overshoot phenomena do not occur: if for instance an agent is allowed to jump but not towards the center of the cluster it only moves little by little and stops when movement negatively impacts convergence. 
 However, simulation runtimes are significantly longer in the continuous case as we need to perform an extra loop of iterations within each unit time-interval.

\begin{figure}[tbhp]
    \centering
    \subfloat[]{\label{fig: Continuous_convergenceTime_VS_n_interp}\includegraphics[width=0.6\textwidth]{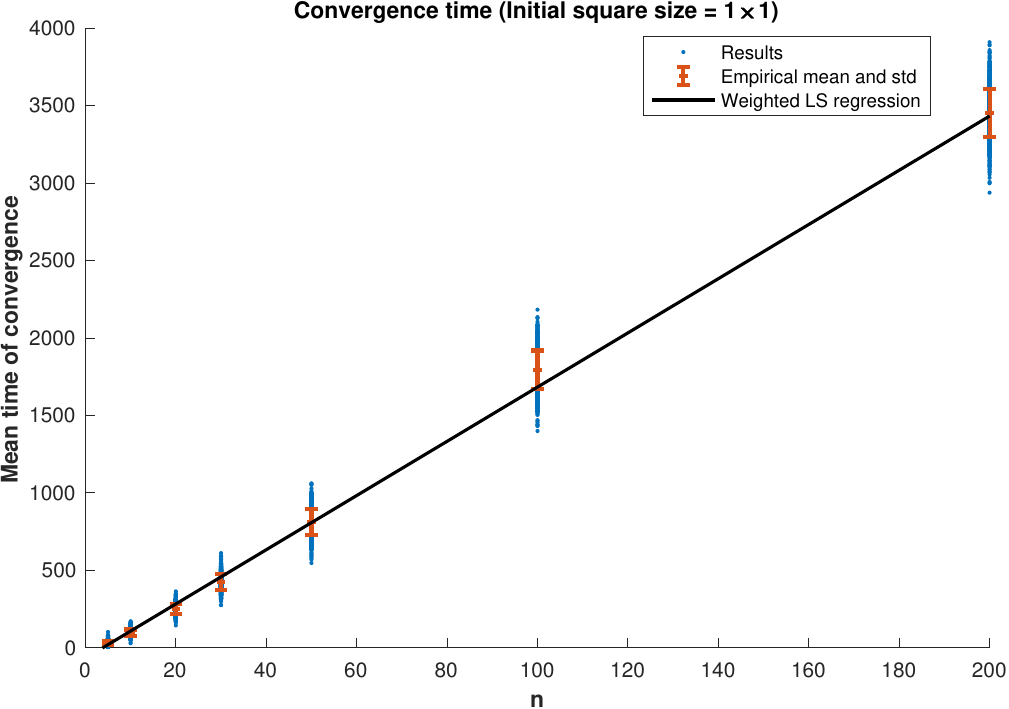}}\\
    \subfloat[]{\label{fig: Continuous_convergenceTime_VS_n_res_bound}\includegraphics[width=0.6\textwidth]{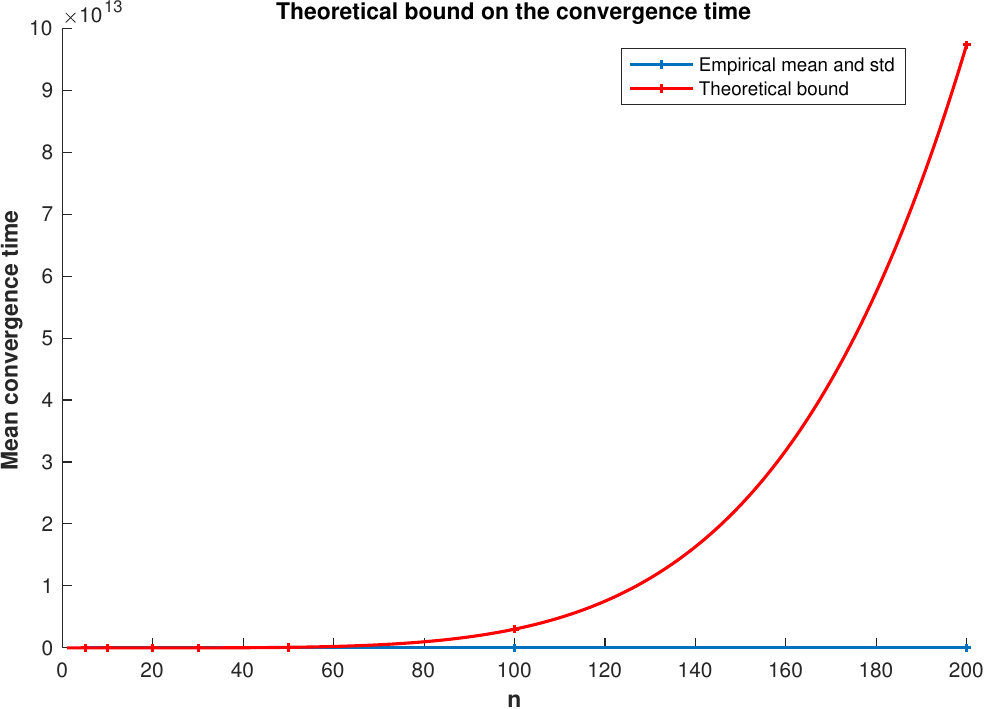}}\\
    \caption{Analysis of the influence of the number of agents $n$ on the convergence time. All simulations were set to an initial random spread of $n$ agents in a $1$ by $1$ area with agents' blind-zone of radius $\delta = 0.02$. Convergence time was taken for the first time when all agents were gathered in a circle of radius $\delta$. This process ran for $1{,}000$ repetitions with different random initial constellations. \Cref{fig: Continuous_convergenceTime_VS_n_interp}: convergence  time vs. $n$, on which we superimpose the empirical mean and standard deviation of the results. The linear graph was obtained using weighted linear least squares fitting to the average results, with weights equal to the squared inverse of the empirical standard deviation of the data to take into account the non uniform variance of the results. \Cref{fig: Continuous_convergenceTime_VS_n_res_bound}: convergence time vs. $n$ and with the bound given in \cref{eq: E time conv small delta explicit}.}
    \label{fig:continuous n analysis curves}
\end{figure}



\section{Conclusion}

We proposed and analyzed two randomized gathering processes for identical, anonymous, oblivious mobile agents, only capable to sense the presence of other agents behind their motion direction. The agents act synchronously, and at unit time-intervals they randomly select new forward motion orientations. We proved that the ``continuous version" of the process ensures gathering to within a region of diameter $2\delta$ where $\delta$ is a parameter setting a ``blind-zone" in sensing nearby agents. Gathering happens in finite expected time, proportional to $\delta^{-1}$. This result also suggests that the ``blind-zone'' is absolutely necessary for finite expected time convergence.

The fully discrete model, in which agents perform unit jumps forward if no agents are detected behind them, was also found empirically  to gather the agents to a minimal enclosing circle of radius randomly varying around $\frac{1}{2}$, in time proportional to the number of agents. This happens in all cases we tested, however a proof of this result has not yet been found and will certainly involve probabilistic convergence arguments. 

We are currently investigating the dynamics of the randomly wandering cluster of agents once gathering has been achieved. We expect to show that the centroid of the system then performs a random walk, or at least undertakes unbiased dynamics. If such is the case, we may be able to guide the swarm by extending our previous work \cite{barel2018steering} were we presented an algorithm for externally controlled steering of swarms of indistinguishable agents, in spite of the agents' lack of information on absolute location and orientation.

\bigskip
\begin{center}
    \footnotesize APPENDICES
\end{center}

\begin{appendices}

\section{Proof of \cref{NeverLoseFriendCorollary}}
\label{App: Proof Bourbaki Never Lose Friend}
Assume $d_{i,j}(t_1)  \le \delta$ with $i\neq j$. Assume that there is $t_2>t_1$ such that $d_{i,j}(t_2)>\delta$. By continuity of $d_{i,j}$ thanks to the continuity of displacements, the intermediate value theorem gives us $t_3\in[t_1,t_2[$ such that $d_{i,j}(t_3) = \delta$. In particular, we can take $t_3 = \sup \{ \tilde{t} \mid \tilde{t}\in[t_1,t_2[ \textrm{ and } d_{i,j}(\tilde{t}) = \delta \}$. By continuity of $d_{i,j}$ we get $d_{i,j}(t_3) = \delta$. If there is a value of $\tilde{t}\in]t_3,t_2[$ such that $d_{i,j}(\tilde{t})\le\delta$, then we can once again apply the intermediate value theorem and get $t_4\in ]t_3,t_1[$ such that $d_{i,j}(t_4) = \delta$ which contradicts the maximality assumption on $t_3$. Therefore ${d_{i,j}}_{|]t_3,t_1[}>\delta$. 
We then have, using \cref{NeverLoseFriend}, that the right derivative of ${d_{i,j}}$ is negative on $[t_3,t_1[$. For any $t_5 \in ]t_3,t_1[ $, we have $d_{i,j}(t_5) \le d_{i,j}(t_3) = \delta$, which is a contradiction since ${d_{i,j}}_{|]t_3,t_1[}>\delta$.
\hfill\begin{tikzpicture}\draw (0,0) rectangle (0.15,0.225); \end{tikzpicture}

\section{Proof of \cref{eq: Sh_min pos}}
\label{App: Proof of sh_min positive}
Consider the case $\delta\tan{\frac{\pi}{4n}} \ge1$. This implies that $Step_{min} = \delta\tan{\frac{\pi}{4n}}$. We thus have
\begin{align}
	Sh_{min}>0 &\iff \tan^2{\frac{\pi}{4n}} -2\tan{\frac{\pi}{4n}}\sin{\frac{\pi}{2n}} < 0 \nonumber \\
	&\iff \tan{\frac{\pi}{4n}}-2\sin{\frac{\pi}{2n}}< 0 \label{eq: tan minus 2 sin negative} \\
	&\iff \frac{\sin{\frac{\pi}{4n}}}{\cos{\frac{\pi}{4n}}\times 2\cos{\frac{\pi}{4n}}\sin{\frac{\pi}{4n}}} < 2 \nonumber \\
	&\iff \cos^2{\frac{\pi}{4n}} > \frac{1}{4} \nonumber \\
	&\iff \frac{\pi}{4n} < \frac{\pi}{3} \nonumber \\
	&\iff n>\frac{3}{4}  \label{eq: n bigger than 3 quarters}
\end{align}
which is always true since trivially $n\ge1$.

Now consider the opposite case when $1<\delta\tan{\frac{\pi}{4n}}$. Thus $Step_{min} = 1$ and then
\begin{equation}
	\label{eq: sh min vdt smaller tan}
	Sh_{min}>0 \iff 1-2\delta\sin{\frac{\pi}{2n}} <0.
\end{equation}

However we know that by assumption $1<\delta\tan{\frac{\pi}{4n}}$, which thus gives
$${1-2\delta\sin{\frac{\pi}{2n}} < \delta(\tan{\frac{\pi}{4n}}-2\sin{\frac{\pi}{2n}})}.$$
According to \cref{eq: tan minus 2 sin negative,eq: n bigger than 3 quarters} and since $n\ge1$, we have that \cref{eq: sh min vdt smaller tan} is always true.

We have thus proved that in all cases $Sh_{min}$ is a strictly positive constant.
\hfill\begin{tikzpicture}\draw (0,0) rectangle (0.15,0.225); \end{tikzpicture}

\end{appendices}


\bibliographystyle{siamplain}
\bibliography{references}

\clearpage

\begin{center}
\textbf{\large SUPPLEMENTARY MATERIALS: PROBABILISTIC GATHERING OF AGENTS WITH SIMPLE SENSORS}
\end{center}
\setcounter{equation}{0}
\setcounter{figure}{0}
\setcounter{table}{0}
\setcounter{page}{1}
\setcounter{section}{0}
\makeatletter
\renewcommand{\theequation}{S\arabic{equation}}
\renewcommand{\thefigure}{S\arabic{figure}}

\section{Analysis of the number of runs necessary for expected convergence time estimation}
It is crucial to understand how reliable an estimator is when performing estimations. In this supplementary material, we analyse the choice of random independent iterations needed to estimate the expected time of convergence of the swarms. We believe an estimation is correct when the empirical average has stabilised, in a qualitative sense, to convergence as enforced by the law of large numbers.

\subsection{Discrete dynamics}
\label{an subsec: Discrete analysis nb runs}

We denote $t_{cv,i}^{Dis}$ the empirical convergence time of trial $i$, which can be seen as the $i$-th independent realisation of the random variable giving the convergence time $t_{cv}^{Dis}$, and
\begin{equation*}
    \hat{t}_{cv,k}^{Dis} = \frac{1}{k}\sum\limits_{i=1}^k t_{cv,i} ^{Dis}
\end{equation*}
the estimated expected convergence time using the first $k$ independent trials. We analyse $\hat{t}_{cv,k}^{Dis}$ as a function of $k$ and work for the final estimation with a number $k$ for which we have qualitatively reached convergence. Furthermore, we also study the evolution of the distribution of realisations of $t_{cv}^{Dis}$ depending on the number of trials so that for our choice of number of trials this distribution has also qualitatively converged.

Analysis of the evolution of $\hat{t}_{cv,k}^{Dis}$ with respect to $k$ is done in \cref{fig: Discrete_convergenceTime_VS_cumsum_trials}. Empirically, the average convergence time stabilises after a few hundred rounds and has reached convergence using $1{,}000$ trials. Analysis of the evolution of the distributions of $t_{cv}^{Dis}$ depending on the number of trials is done in \cref{fig:discrete histograms n}. The normalised distributions have qualitatively reached convergence using $1{,}000$ trials. We thus decided to use $1{,}000$ simulations in order to estimate average convergence times. Furthermore, the distribution of the convergence time around the mean converges to a symmetric Gaussian-like curve, implying that we can fully summarize the empirical distributions at convergence with the empirical mean and with the traditional empirical unbiased standard deviation estimator $\hat{\sigma}_k^{Dis}$
\begin{equation*}
    \hat{\sigma}_k^{Dis} = \sqrt{\frac{1}{n-1}\sum\limits_{i=1}^k(t_{cv,i}^{Dis} - \hat{t}_{cv,k}^{Dis})}
\end{equation*}
choosing $k = 1{,}000$.

\begin{figure}[htbp]
  \centering
    \includegraphics[width=0.75\textwidth]{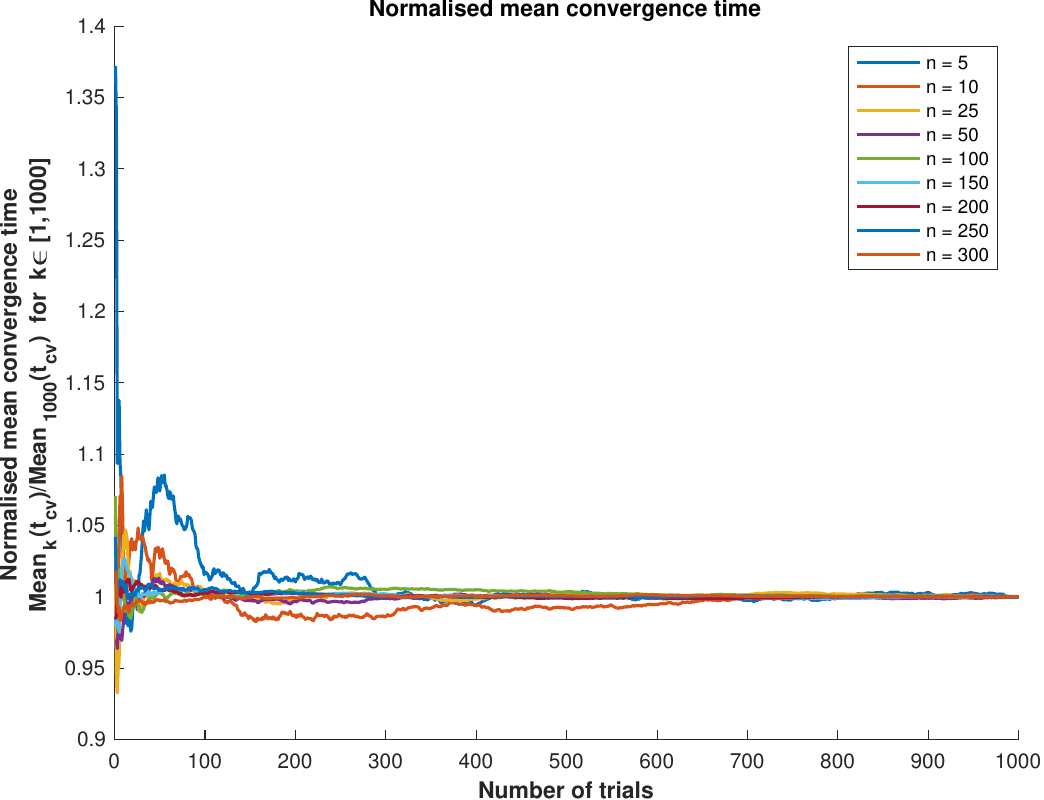}
    \caption{Normalised mean convergence time vs. the number of trials used for computing the mean. Normalisation is done by dividing the mean values by the mean value using all the data (1{,}000 trials). All simulations were set to an initial random spread in a $50$ by $50$ area with unit step-size. Convergence time was taken for the first time when all agents were gathered in a circle of radius $1$.}
      \label{fig: Discrete_convergenceTime_VS_cumsum_trials}
\end{figure}

\begin{figure}[tbhp]
    \centering
    \subfloat[$n=5$]{\label{fig: Dicrete_hist_n_5}\includegraphics[width=0.4\textwidth]{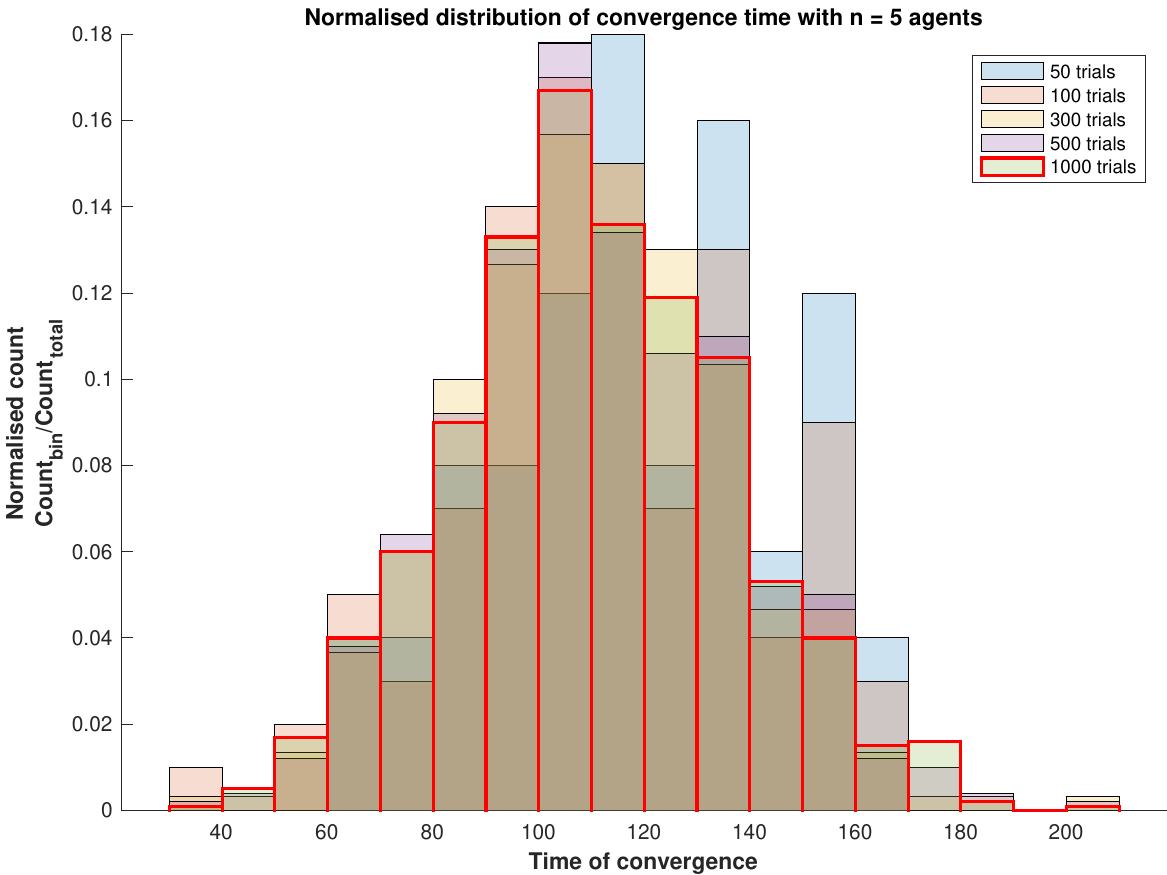}}~\quad
    \subfloat[$n=50$]{\label{fig: Dicrete_hist_n_50}\includegraphics[width=0.4\textwidth]{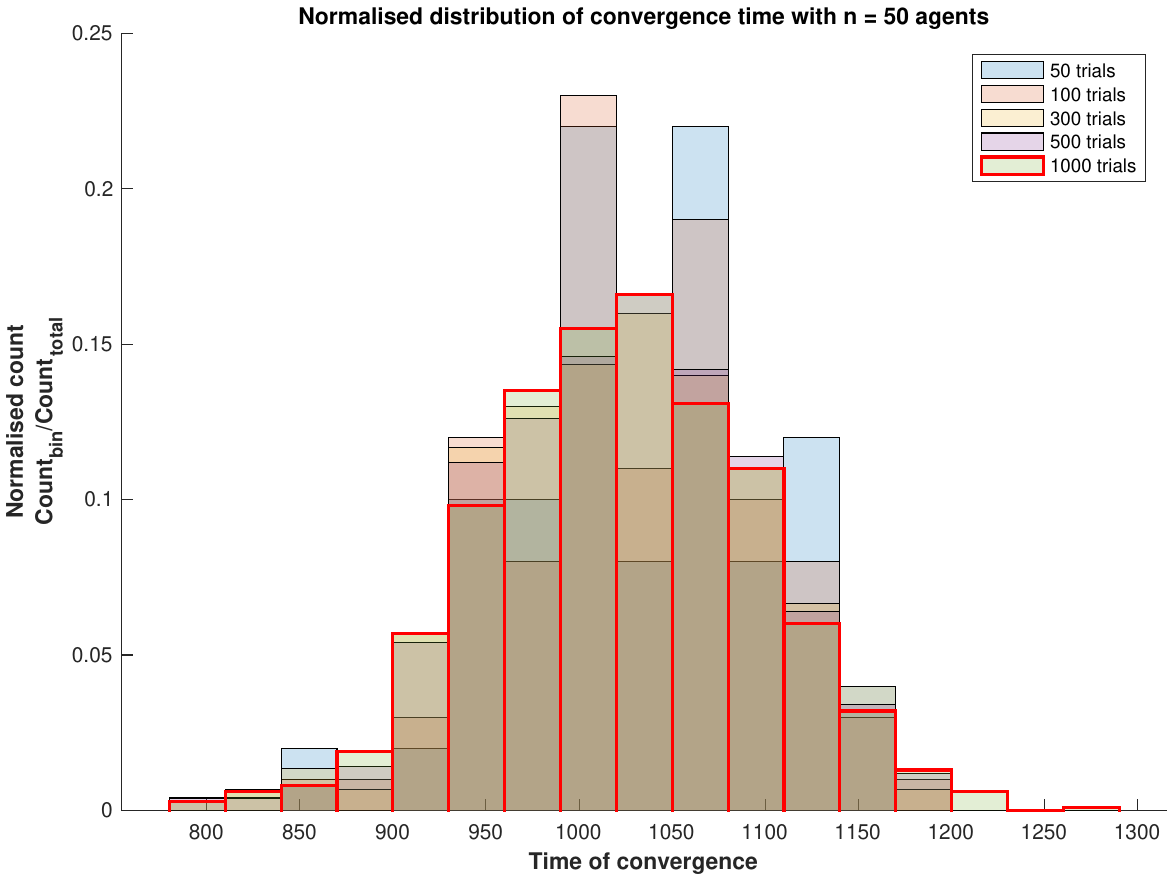}}\\
    \subfloat[$n=150$]{\label{fig: Dicrete_hist_n_150}\includegraphics[width=0.4\textwidth]{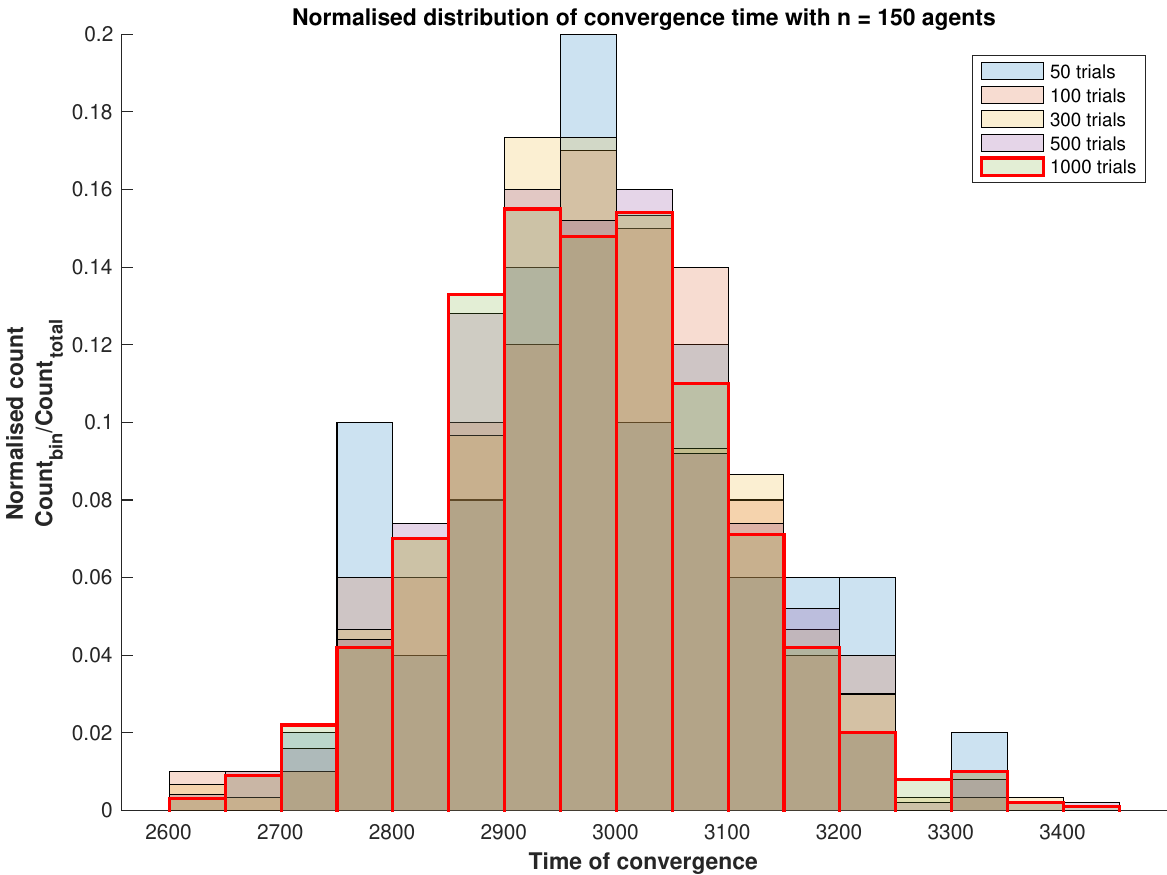}}
    \caption{Histograms of the convergence time using different number of trials for various number of agents $n$. The histograms are normalised by their total number of counts (the number of trials) for comparison. Simulation were set to an initial random spread of $n$ agents in a $50$ by $50$ area with unit step-size.}
    \label{fig:discrete histograms n}
\end{figure}

\subsection{Continuous dynamics}
\label{an subsec: Continuous analysis nb runs}

Similarly to the discrete case, we denote $t_{cv,i}^{Con}$ the empirical convergence time of trial $i$, which can be seen as the $i$-th independent realisation of the random variable giving the convergence time $t_{cv}^{Con}$, and
\begin{equation*}
    \hat{t}_{cv,k}^{Con} = \frac{1}{k}\sum\limits_{i=1}^k t_{cv,i} ^{Con}
\end{equation*}
the estimated expected convergence time using the first $k$ independent trials. We perform the same analysis of $\hat{t}_{cv,k}^{Con}$ and of the distribution of the realisation $t_{cv}^{Con}$ as in the discrete case.

Analysis of the evolution of $\hat{t}_{cv,k}^{Con}$ with respect to $k$ is done in \cref{fig:Continuous_convergenceTime_VS_cumsum_trials_fixed_n_then_delta}. Empirically, the average convergence time stabilises after a few hundred rounds and has reached convergence using $1{,}000$ trials. Analysis of the evolution of the distributions of $t_{cv}^{Con}$ depending on the number of trials is done in \cref{fig:continuous histograms n and delta}. The normalised distributions have qualitatively reached convergence using $1{,}000$ trials. We thus decided to use $1{,}000$ simulations in order to estimate average convergence times. Furthermore, the distribution of the convergence time around the mean converges to a symmetric Gaussian-like curve, implying that we can fully summarize the empirical distributions at convergence with the empirical mean and with the traditional empirical unbiased standard deviation estimator $\hat{\sigma}_k^{Con}$
\begin{equation*}
    \hat{\sigma}_k^{Con} = \sqrt{\frac{1}{n-1}\sum\limits_{i=1}^k(t_{cv,i}^{Con} - \hat{t}_{cv,k}^{Con})}
\end{equation*}
choosing $k = 1{,}000$.

\begin{figure}[tbhp]
    \centering
    \subfloat[$n=10$]{\label{fig: Continuous_convergenceTime_VS_cumsum_trials_fixed_n}\includegraphics[width=0.75\textwidth]{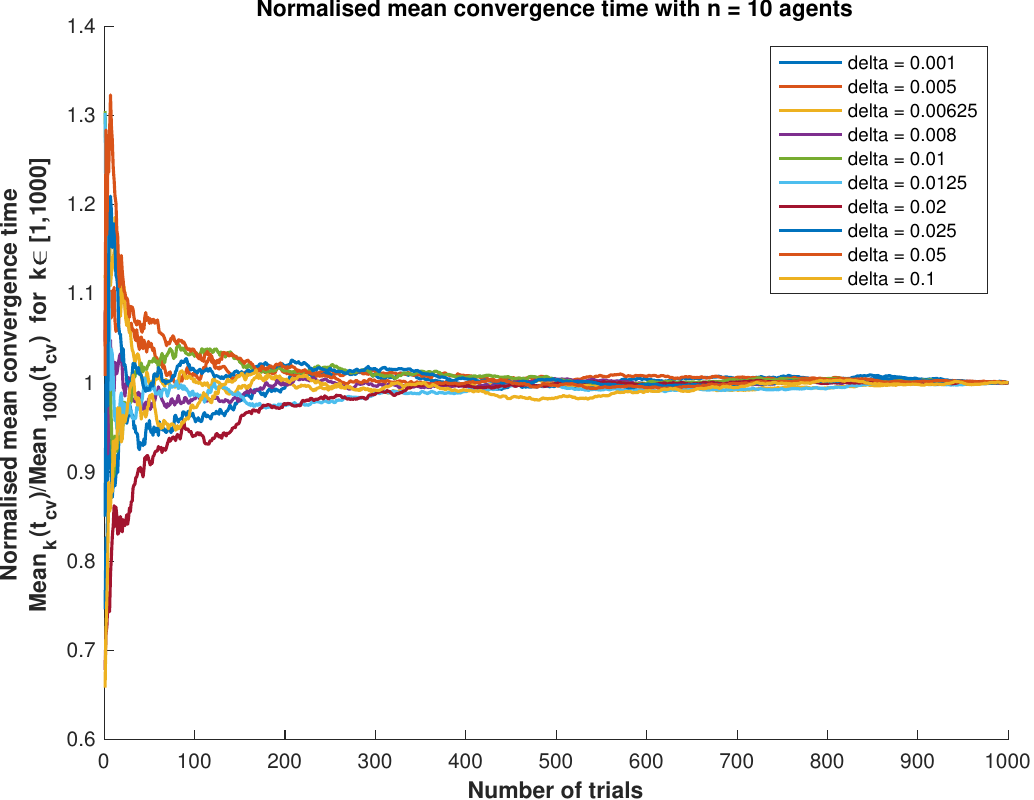}} \\
    \subfloat[$\delta=0.02$]{\label{fig: Continuous_convergenceTime_VS_cumsum_trials_fixed_delta}\includegraphics[width=0.75\textwidth]{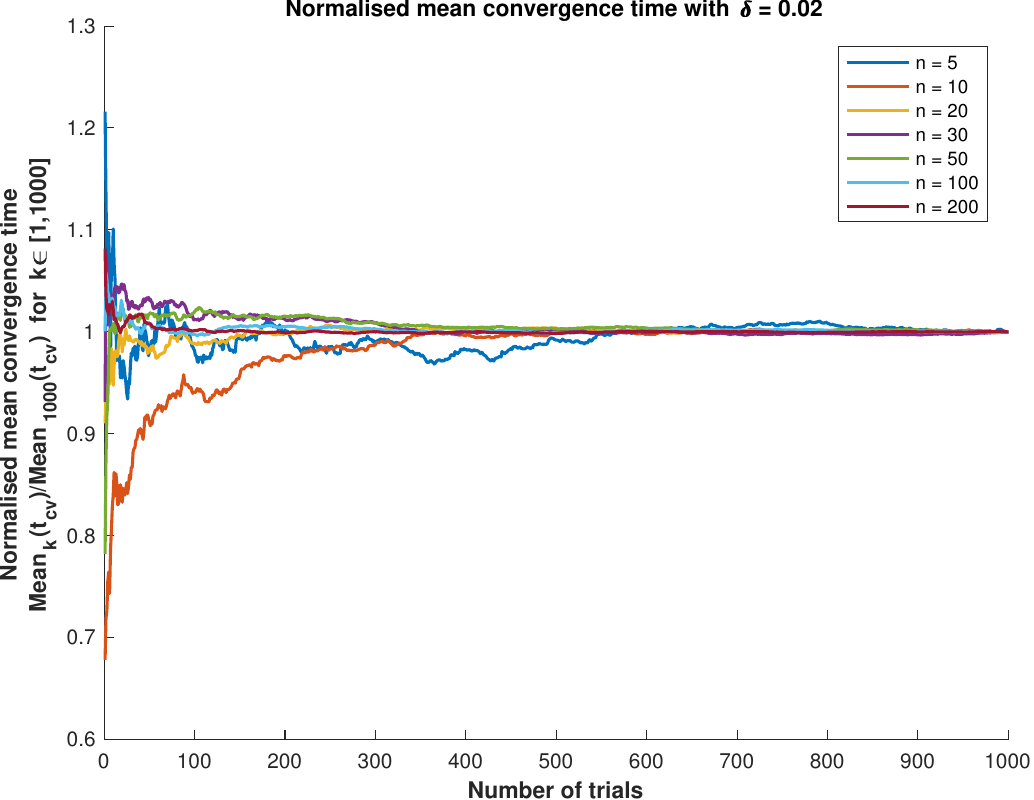}}
    \caption{Normalised mean convergence time vs. the number of trials used for computing the mean for the piecewise continuous dynamics algorithm. Normalisation is done by dividing the mean values by the mean value using all the data (1{,}000 trials). Convergence time was taken for the first time when all agents were gathered in a circle of radius $\delta$. In \cref{fig: Continuous_convergenceTime_VS_cumsum_trials_fixed_n}, all simulations were set to an initial random spread of a fixed number of $10$ agents in a $1$ by $1$ area. In \cref{fig: Continuous_convergenceTime_VS_cumsum_trials_fixed_delta}, all simulations were set to an initial random spread in a $1$ by $1$ area with agents' blind-zone of fixed radius $\delta = 0.02$.}
    \label{fig:Continuous_convergenceTime_VS_cumsum_trials_fixed_n_then_delta}
\end{figure}

\begin{figure}[tbhp]
    \centering
    \subfloat[$n=5$, $\delta=0.02$]{\label{fig: Continuous_hist_n_5_delta_0.02}\includegraphics[width=0.45\textwidth]{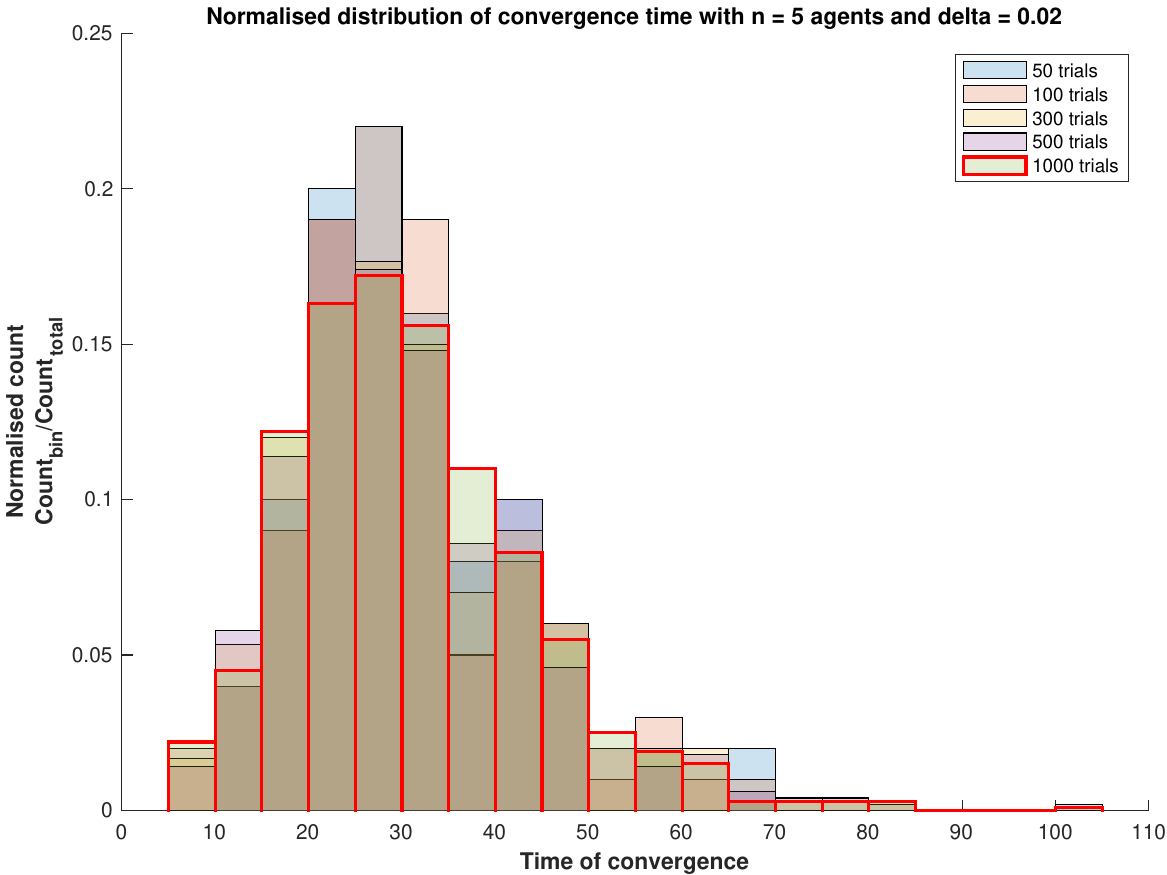}}~\quad
    \subfloat[$n=10$, $\delta=0.02$]{\label{fig: Continuous_hist_n_10_delta_0.02}\includegraphics[width=0.45\textwidth]{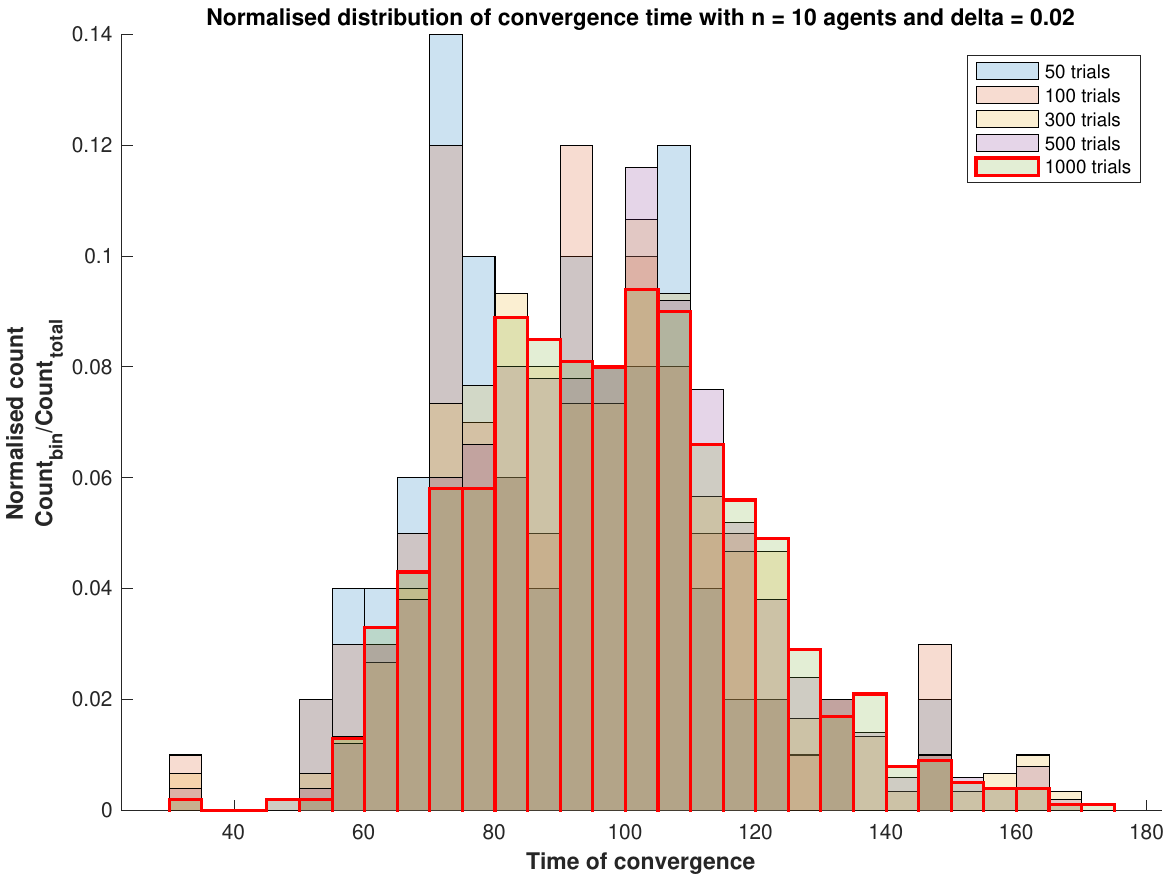}}\\
    \subfloat[$n=50$, $\delta=0.02$]{\label{fig: Continuous_hist_n_50_delta_0.02}\includegraphics[width=0.45\textwidth]{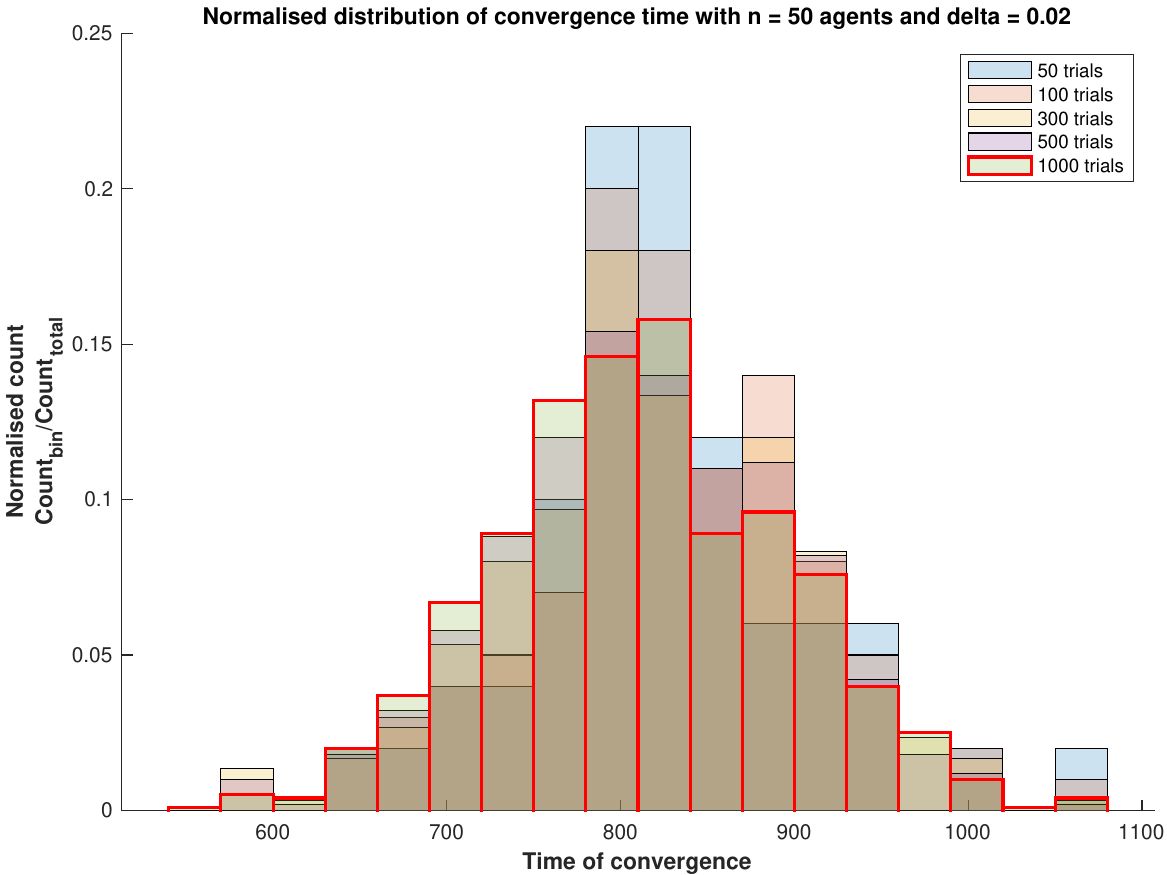}}~\quad
    \subfloat[$n=200$, $\delta=0.02$]{\label{fig: Continuous_hist_n_200_delta_0.02}\includegraphics[width=0.45\textwidth]{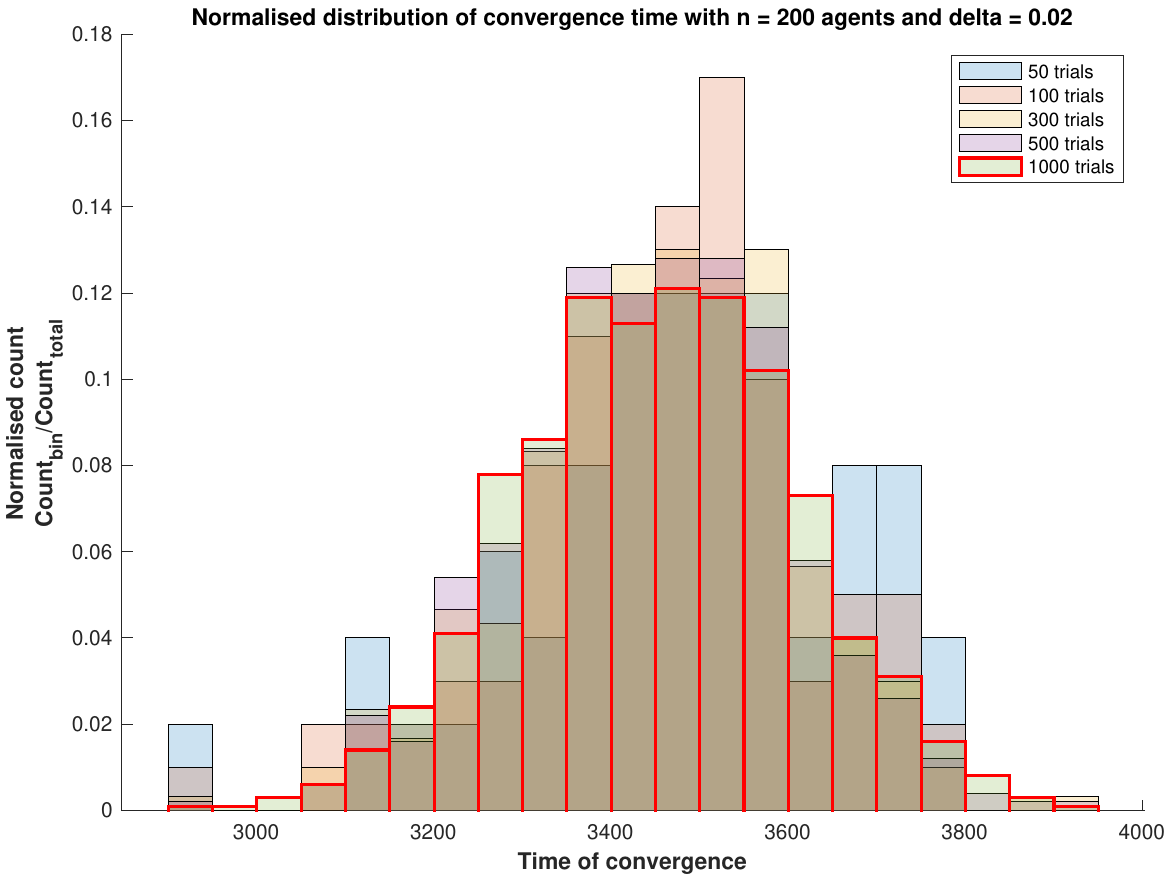}}\\
    \subfloat[$n=10$, $\delta=0.005$]{\label{fig: Continuous_hist_n_10_delta_0.005}\includegraphics[width=0.45\textwidth]{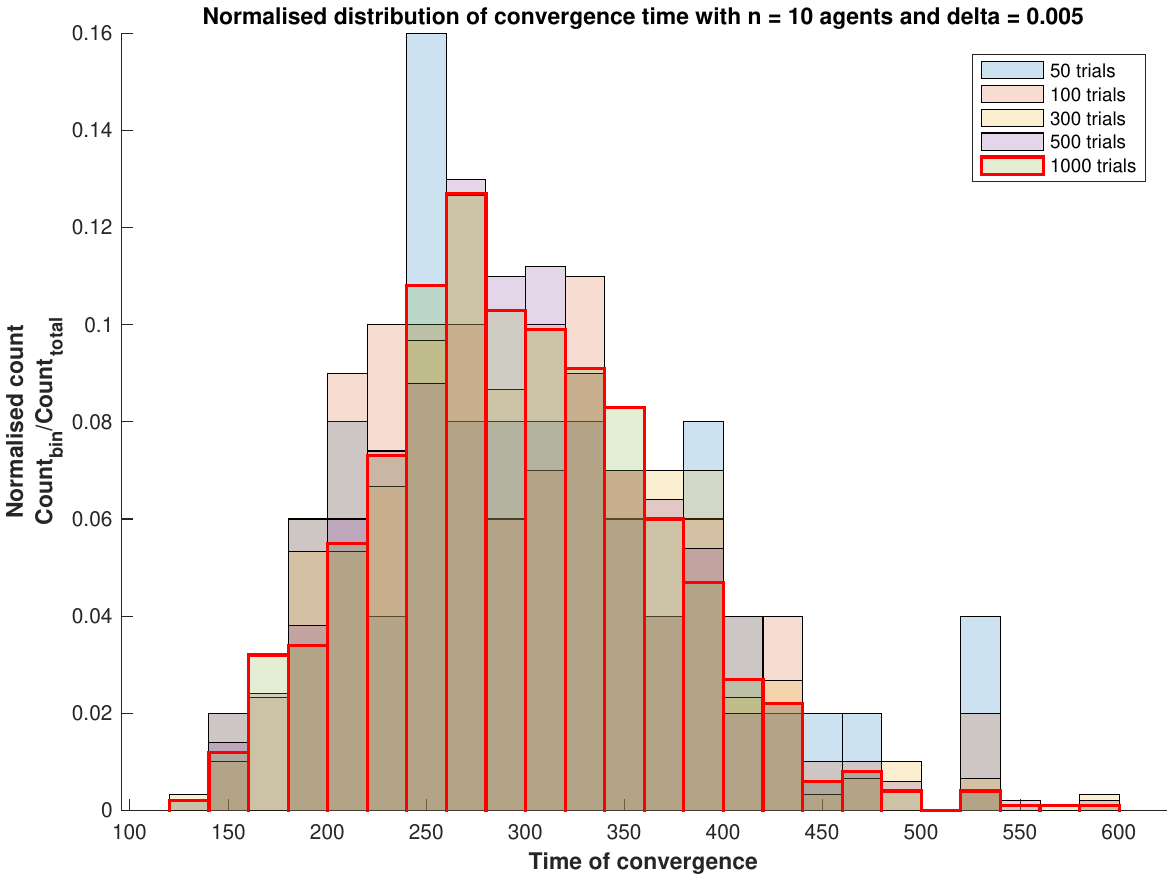}}~\quad
    \subfloat[$n=10$, $\delta=0.001$]{\label{fig: Continuous_hist_n_10_delta_0.001}\includegraphics[width=0.45\textwidth]{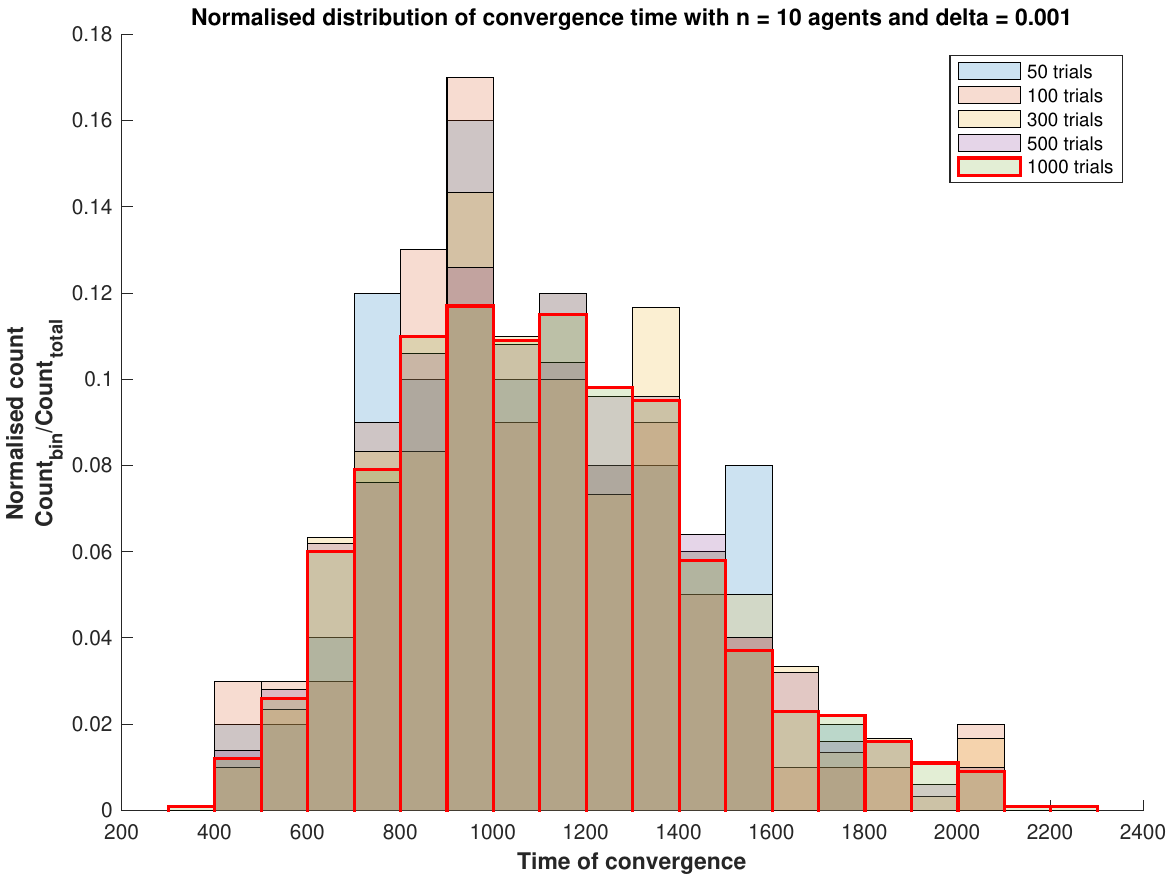}}
    \caption{Histograms of the convergence time using different number of trials for the piecewise continuous dynamics algorithm. The histograms are normalised by their total number of counts (the number of trials) for comparison. Simulations in each case were set to an initial random spread, in a $1$ by $1$ area, of $n$ agents with agents' blind-zone of radius $\delta$. The smaller step-size $dt$ is $0.01$ (respectively $0.002$ and $0.0005$) for $\delta$ equal to $0.02$ (respectively $0.005$ and $0.001$).}
    \label{fig:continuous histograms n and delta}
\end{figure}

\end{document}